\documentclass[conference]{IEEEtran}
\IEEEoverridecommandlockouts
\usepackage{cite}
\usepackage{amsmath,amssymb,amsfonts}
\usepackage{algorithmic}
\usepackage{graphicx}
\usepackage{textcomp}
\usepackage{xcolor}
\usepackage{anyfontsize}
\usepackage{mathrsfs,amsmath,url,amsthm,color,comment}
\usepackage{hyperref}
\usepackage{thm-restate}
\usepackage{todonotes}
\usepackage[nameinlink,capitalise]{cleveref}

\theoremstyle{plain}

    \declaretheorem[name=Theorem]{theorem}
    \declaretheorem[name=Proposition,numberlike=theorem]{proposition}
  \declaretheorem[name=Corollary,numberlike=theorem]{corollary}
  \declaretheorem[name=Conjecture,numberlike=theorem]{conjecture}
  \declaretheorem[name=Claim,numberlike=theorem]{claim}
  \declaretheorem[name=Observation,numberlike=theorem]{observation}
  
  \declaretheorem[name=Lemma,numberlike=theorem]{lemma}

\theoremstyle{definition}
\declaretheorem[name=Definition,numberlike=theorem]{definition}
     \declaretheorem[name=Example,numberlike=theorem]{example}

 \newcommand{\ignore}[1]{}

\DeclareMathOperator{\Csp}{CSP}
\DeclareMathOperator{\CSP}{CSP}
\DeclareMathOperator{\Pol}{Pol}
\DeclareMathOperator{\Aut}{Aut}

\DeclareMathOperator{\proj}{proj}
\DeclareMathOperator{\V}{\mathcal{V}}
\DeclareMathOperator{\constraints}{\mathcal{C}}
\DeclareMathOperator{\instance}{\mathcal{I}}

\DeclareMathOperator{\scope}{S}

\newcommand{\inj}{\textrm{inj}}

\newcommand\injinstance[1][\instance]{#1^{(\inj)}}
\newcommand\injinstances{\textrm{Inj}}

\newcommand{\chmichal}[1]{{#1}}
\newcommand{\cht}[1]{{#1}}
\newcommand{\michal}[1]{\todo[linecolor=olive,backgroundcolor=olive!25,bordercolor=olive]{\textbf{Micha\l:} #1}}

\newcommand\mic[1]{{#1}}

\newcommand\rel[1]{\mathbb{#1}}

\newcommand{\tuple}[1]{\mathbf{#1}}

\newcommand{\relstr}[1]{\mathbb{#1}}

\newcommand{\sA}{\mathbb A}
\newcommand{\sB}{\mathbb B}
\newcommand{\sC}{\mathbb C}

\newcommand{\sF}{\mathbb F}

\newcommand{\sK}{\mathbb K}

\newcommand\sX{\mathbb X}

\newcommand{\fG}{\mathscr G}

\newcommand{\cJ}{\mathcal J}
\newcommand{\gG}{\mathscr G}

\begin{document}

\title{New  Sufficient Algebraic  Conditions for Local Consistency over Homogeneous  Structures of Finite Duality\\
\thanks{This research was funded in whole or in part by the Austrian Science Fund (FWF) [P 32337, I 5948]. Partially funded by the National Science Centre, Poland under the Weave program grant no. 2021/03/Y/ST6/00171 and by the grant number 2020/37/B/ST6/01179.. For the purpose of Open Access, the authors have applied a CC BY public copyright licence to any Author Accepted Manuscript (AAM) version arising from this submission. This research is  funded by the European Union (ERC, POCOCOP, 101071674). Views and opinions expressed are however those of the author(s) only and do not necessarily reflect those of the European Union or the European Research Council Executive Agency. Neither the European Union nor the granting authority can be held responsible for them.}
}

\author{\IEEEauthorblockN{Tom\'a\v{s} Nagy}
\IEEEauthorblockA{\textit{Theoretical Computer Science Department} \\
\textit{Jagiellonian University}\\
Krakow, Poland \\
https://orcid.org/0000-0003-4307-8556}
\and
\IEEEauthorblockN{Michael Pinsker}
\IEEEauthorblockA{\textit{Institute of Discrete Mathematics and Geometry (Algebra)} \\
\textit{Technische Universit\"{a}t Wien}\\
Vienna, Austria \\
https://orcid.org/0000-0002-4727-918X}
\and
\IEEEauthorblockN{Micha{\l} Wrona}
\IEEEauthorblockA{\textit{Theoretical Computer Science Department}\\ \textit{Jagiellonian University} \\
Krak\'{o}w, Poland \\
https://orcid.org/0000-0002-2723-0768}}

\maketitle

\begin{abstract}
The path to the solution of Feder-Vardi dichotomy conjecture by Bulatov and Zhuk led through showing that more and more general algebraic conditions imply polynomial-time algorithms for the finite-domain Constraint Satisfaction Problems (CSPs) whose templates satisfy them.
These investigations resulted in the discovery of the appropriate height~1 Maltsev conditions characterizing bounded strict width, bounded width, the applicability of the few-subpowers algorithm, and many others.

\mic{For problems in the range of the similar Bodirsky-Pinsker conjecture on infinite-domain CSPs, one can only find such a characterization for the notion of bounded strict width, with a proof essentially the same as in the finite case. In this paper, we provide the first non-trivial results showing that certain height~1 Maltsev conditions imply bounded width, and in consequence tractability, for 
a natural subclass of templates within the Bodirsky-Pinsker conjecture which includes many templates in the literature as well as templates for which no complexity classification is known.}

\end{abstract}
%


\section{Introduction}

This paper treats the complexity of \emph{Constraint Satisfaction Problems}  
parametrized by finite and infinite relational structures \mic{$\relstr A = (A; R_1,\ldots)$},
 called \emph{templates}. Every $\sA$
gives rise to a computational decision  problem $\Csp(\sA)$ whose instances consist of 
 a set of variables and a set of constraints   
  which restrict the  
possible combinations of  values the variables can take. More precisely, in a solution to the instance the variables have to be assigned values in the \emph{domain} of $\sA$, i.e.~its base set;  the constraints require some tuples of variables to belong to some relations of $\sA$ after the assignment of values. The problem generalizes the well-known problem  $3$-SAT where 
the template $\sA$ over the domain $A = \{0, 1\}$   contains  exactly one ternary relation for every type of clause over three propositional literals. Furthermore, \mic{the} $k$-coloring \mic{problem} as well as the question of whether a given system of linear equations over a finite field has a solution \mic{can} 
be expressed naturally as $\Csp(\sA)$ for a finite structure $\sA$. The diversity of problems that fall under this formalism and certain preliminary results led Tom\'{a}s Feder and Moshe Vardi~\cite{FederVardi}  to formulate their famous \emph{dichotomy conjecture} \mic{stating} 
 that all \mic{finite-domain} 
 CSPs are either in P or 
  NP-complete. \mic{The conjecture} 
   contrasts the result known as Ladner's theorem~\cite{Ladner} which asserts that P$\neq$NP implies the existence of  NP-intermediate problems. 
No serious progress \mic{had} 
 been \mic{achieved} 
  towards the resolution of the conjecture until the so-called \emph{universal-algebraic approach}
\mic{was} 
 employed~\cite{Jeavsons98,BJK05} and
 the perspective switched  from relational structures $\relstr A$ as  central objects first to algebras of multivariate functions on $\relstr A$ leaving the relations of $\relstr A$ invariant \mic{(so-called \emph{polymorphisms}   of $\sA$, the set of which is denoted by $\Pol(\sA)$)}, 
  and then \mic{even further} to identities \mic{satisfied} 
   by these algebras 
   \mic{(equivalently, to so-called \emph{algebraic varieties})}. Although 
   \mic{rather} different, both resolutions to the Feder-Vardi conjecture, \mic{the one due to Andrei Bulatov~\cite{Bulatov:2017} as well as  the one due to  Dmitriy Zhuk~\cite{Zhuk:2017,Zhuk:2020}, confirm the (finite-domain) \emph{algebraic tractability conjecture} from~\cite{BJK05}  using universal-algebraic methods, thereby  demonstrating the strength and appropriateness of this approach.}

It is difficult to overstate the importance of the Bulatov-Zhuk dichotomy theorem. Nevertheless, the majority of problems expressible as $\Csp(\sA)$ require an infinite structure $\sA$. 
Fortunately, \mic{for many of them  
the universal-algebraic approach, originally  invented for finite structures, is also applicable.}
 \mic{This is true, in particular, for a certain} class of scheduling problems called \emph{temporal CSPs}~\cite{temporalCSP} (where the relations of templates are given by quantifier-free formulas using the order of the rational numbers), a class of problems in phylogenetic reconstruction \mic{called}  \emph{phylogeny  CSPs}~\cite{PhylogenyCSP} (where suitable templates can be provided using quantifier-free formulas using the \emph{homogeneous C-relation}), \mic{for} \emph{Graph-SAT problems}~\cite{Schaefer-Graphs,HomogeneousGraphs} (where templates have relations given by quantifier-free formulas over countable homogeneous graphs) and, more generally,   \emph{Hypergraph-SAT problems}~\cite{hypergraphs}, for certain  \emph{completion problems}~\cite{BodPro25,feller2024algebraic,BitterM24} (where templates have relations given by quantifier-free formulas over countable homogenous digraphs), as well as for the model-checking problem for sentences in the logic MMSNP of Feder and Vardi~\cite{MMSNP,MMSNP-journal} or the logic  GMSNP~\cite{GMSNP,Bienvenu:2014}  (for which suitable infinite-domain templates can be provided using non-trivial model-theoretic techniques). 
A conjecture similar to \mic{the} Feder-Vardi conjecture has been formulated for \mic{a class of infinite structures sufficiently large to contain the templates of these problems as well as all finite-domain CSPs, and sufficiently small to be contained in NP and allow for a universal-algebraic approach.} 

\begin{conjecture} (Bodirsky-Pinsker \mic{2011})
\label{conj:BodirskyPinsker}
 Let $\mathbb{A}$ be the model-complete core of a finite-signature structure which is first-order definable in a finitely bounded homogeneous structure. 
 Then
precisely one of the following holds:
\begin{enumerate}
\item  \label{conj:projecthomo} there exists an expansion $\mathbb{A}'$
of $\mathbb{A}$ by finitely many constants such that $\Pol(\mathbb{A}')$ has a
continuous projective homomorphism (and \mic{hence} $\CSP(\mathbb{A})$ is NP-complete);
\item \label{conj:noprojecthomo} for any such expansion $\mathbb{A}'$ 
the clone $\Pol(\mathbb{A}')$ has no continuous projective homomorphism. In this case, $\CSP(\mathbb{A})$ is in P.
\end{enumerate}
\end{conjecture}

The above 
 original formulation of the infinite-domain dichotomy conjecture restricts the infinite structures under consideration to so-called \emph{model-complete cores}, whose exact definition is not of importance for this paper -- this restriction of the scope of the conjecture is without loss  of generality, and we refer the interested reader to~\cite{BodirskyBook}.  
  All notions which are crucial for understanding the merits of the present  paper are defined in Section~\ref{sect:Prelims}. 
 
 The \mic{first item of the conjecture indeed implies NP-hardness}.  
  The tractability part, however, stays 
  wide open: Condition~(\ref{conj:noprojecthomo}) is just the negation of Condition~(\ref{conj:projecthomo}), and therefore does not give any substantial information which could set the direction for developing a polynomial-time algorithm.
Fortunately, 
a result~\cite{BartoPinskerDichotomy} due to Libor Barto and \mic{Michael Pinsker} 
 states that 
Condition~(\ref{conj:noprojecthomo}) is equivalent to the following \emph{infinite-domain algebraic tractability conjecture}.

\begin{conjecture}(Barto-Pinsker 2016) 
\label{conj:inftract}
Let   $\sA$ be the model-complete core of a  finite-signature structure which is first-order definable in a finitely bounded homogeneous structure.
If $\sA$ is invariant under a \emph{pseudo-Siggers} operation, i.e., a 6-ary operation $s$ for which there are unary polymorphisms $\alpha, \beta$
such that
\begin{equation}
\label{eqn:pseudoSiggers}
\alpha s(x, y, x, z, y, z) \approx \beta s(y, x, z, x, z, y)\nonumber
\end{equation}
for all $x,y,z$ in domain,  then $\Csp(\sA)$ is solvable in polynomial time.
\end{conjecture}

Conjecture~\ref{conj:inftract} is almost identical to the \mic{algebraic tractability criterion for finite-domain CSPs confirmed by} 
 Bulatov and Zhuk,  with the difference that in the finite-domain version 
 $\alpha, \beta$ \mic{are omitted; $s$ is then simply called a \emph{Siggers operation}.} 

Prior to the resolution of the finite-domain algebraic tractability conjecture, 
 research focused on 
 \mic{restrictions thereof obtained by imposing additional assumptions of two kinds:
 \begin{itemize}
 \item directly on the template: e.g.~its size~\cite{Schaefer78,Bulatov02}, or its relations such as the cases of symmetric graphs~\cite{HellN90,Bulatov05} or smooth digraphs~\cite{BartoKN09}
 \item on the algebraic invariants (polymorphisms) of the template: the satisfaction  of  certain  identities by its polymorphisms~\cite{BartoKozikCD,BartoKozikBoundedWidth,IdziakMMVW07}) -- such conditions are called \emph{Maltsev conditions}.
 \end{itemize}
 }
 Indeed, researchers considered identities stronger than \mic{the above-mentioned Siggers identity implying  tractability}, 
  and in particular provided algebraic characterizations of the limits of the applicability of
local consistency checking~\cite{BartoKozikBoundedWidth} as well as of \mic{algorithms resembling  Gaussian elimination}~\cite{IdziakMMVW07}. Both of these algorithmic methods play a crucial role 
 in the final procedures of the algorithms of Bulatov and Zhuk.
It is fair to say that research on \mic{both kinds of restrictions of the conjecture, i.e.~by imposing structural or algebraic assumptions on the template,}  
\mic{not only contributed to its} 
 eventual resolution, 
 \mic{but also} 
  lead to results of independent interest.

  \mic{While work on infinite-domain CSPs has also delivered impressive results, these were, however, almost exclusively confined to the former kind, that is, structural restrictions directly on the template~\cite{temporalCSP, Schaefer-Graphs,MMSNP-journal,hypergraphs}. In fact, there are only two general results guaranteeing polynomial-time solvability of  CSPs within the Bodirsky-Pinsker conjecture under  conditions on their polymorphisms:
   the one  in~\cite{ReductsUnary} 
employs  so-called \emph{canonical pseudo-Siggers operations};  and the one in~\cite{BodirskyDalmau} applies to  templates invariant under a \emph{quasi near-unanimity operation}, i.e. an operation satisfying the identities  
\begin{equation}
\label{eq:nu}
f(y,x, \ldots, x) =  \cdots = f(x,x, \ldots, y) = f(x, \ldots, x)
\end{equation}
for all $x,y$ in the domain. The former provides a  black-box reduction to finite-domain CSPs, while the proof of the latter is a straightforward generalization of the result in~\cite{FederVardi} for finite templates.
}

\mic{Note that the identities in~(\ref{eq:nu}) are of \emph{height~1}, i.e.~there is no nesting of function symbols and variables only appear as arguments in function symbols (they are never isolated).  The prefix \emph{quasi} (as here in \emph{quasi near-unanimity operation}) always refers to the omission of the identity $x=f(x,\ldots,x)$ (which is not of height~1); functions $f$ satisfying this additional identity are called \emph{idempotent}. In consequence, for a \emph{near-unanimity operation} the last term in~(\ref{eq:nu}) could  be replaced by  $x$.  
One might wonder} 
 why we do not compose every side of each of the identities in~(\ref{eq:nu})  with unary operations  to obtain a \emph{pseudo near-unanimity operation} resembling the pseudo-Siggers operation from~\Cref{conj:inftract} \mic{(which is not of height~1)}. The reason is that due to an observation of Marcin Kozik and  \mic{Micha{\l}  Wrona}, \mic{this seemingly weaker algebraic condition turns out to be equivalent} 
 (see Lemma~10.1.3 in~\cite{BodirskyBook}). 

\mic{An algebraic condition 
 of importance for finite-domain templates
is given by chains of \emph{idempotent Pixley} operations, which we here generalize to chains of \emph{quasi Pixley} operations.} The  question whether this condition,  either the original (idempotent) one or its quasi-version,  implies tractability and/or bounded width  over first-order reducts of finitely bounded homogeneous structures   is 
Problem~(35) 
 in Section~14 in~\cite{BodirskyBook}. We solve this question positively for an important subclass of infinite structures within Bodirsky-Pinsker conjecture. 

\begin{definition}
\label{def:QuasiPixley}
A sequence $(P_1, \ldots , P_n)$ of ternary operations on a set $A$ is
called a chain of \emph{quasi Pixley operations} if for all $x, y \in A$

\begin{eqnarray}
P_1(x, y, y) &=& P_1(x, x, x), \nonumber\\
P_i(x, y, x) &=& P_i(x, x, x) \hspace{36 pt } \textrm{ for all } i \in \{1, \ldots , n\}, \nonumber\\
P_i(x, x, y) &=& P_{i+1}(x, y, y) \hspace{30 pt} \textrm{ for all } i \in \{1, \ldots , n\}, \nonumber\\
P_n(x, x, y) &=& P_n(y, y, y).\nonumber
\end{eqnarray}
\end{definition}

\mic{Another algebraic condition of height~1 identities unknown to imply polynomial-time solvability} is given by chains of \emph{quasi directed  J\'{o}nsson} operations. These  are provably implied by  quasi near-unanimity operations: every structure invariant under operations satisfying the latter identities is also invariant under operations satisfying the  former (\mic{see Proposition~6.9.12 in~\cite{BodirskyBook}).}

\begin{definition}
\label{def:QuasiDJonsson}
A sequence $(D_1, \ldots , D_{n})$ of ternary operations on a set $A$ is
called a chain of quasi directed J\'{o}nsson operations if for all $x, y \in A$:
\begin{eqnarray}
\label{eq:D1} D_1(x,x,y) &=& D_1(x,x,x), \nonumber\\
\label{eq:Di} D_i(x,y,x) &=& D_i(x,x,x) \hspace{36 pt }\textrm{ for all } i \in [n],\nonumber\\
\label{eq:Dii+1} D_i(x,y,y) &=& D_{i+1}(x,x,y) \hspace{30 pt} \textrm{ for all } i \in [n-1],\nonumber\\
\label{eq:Dn} D_n(x,y,y) &=& D_n(y,y,y).\nonumber
\end{eqnarray} 
\end{definition}

Finally, the weakest (\mic{see Proposition~6.9.9 in~\cite{BodirskyBook})} algebraic condition of height~1 considered in this paper is the following. 

\begin{definition}
\label{def:QuasiJonsson}
A sequence $(J_1, \ldots , J_{2n+1})$ of ternary operations on a set $A$ is
called a chain of \emph{quasi J\'{o}nsson operations} if for all $x, y \in A$:
\begin{eqnarray}
\label{eq:J1} J_1(x,x,y) &=& J_1(x,x,x), \\
\label{eq:Ji} J_i(x,y,x) &=& J_i(x,x,x) \hspace{18 pt }\textrm{ for all } i \in [2n+1],\\
\label{eq:J2i-12i} J_{2i-1}(x,y,y) &=& J_{2i}(x,y,y) \hspace{18 pt} \textrm{ for all } i \in [n+1],\\
\label{eq:J2i2i+1} J_{2i}(x,x,y) &=& J_{2i+1}(x,x,y) \hspace{10 pt} \textrm{ for all } i \in [n],\\
\label{eq:J2n+1} J_{2n+1}(x,y,y) &=& J_{2n+1}(y,y,y).
\end{eqnarray} 
\end{definition}

The main contribution of this paper, which we shall  explain in detail in the following subsection,  states that the invariance of a large class of  infinite structures  $\sA$ within the Bodirsky-Pinsker conjecture   under a chain of quasi J\'{o}nsson operations implies that the corresponding $\Csp(\sA)$ is solvable by local consistency checking, and is in consequence  in P. Our result thereby lifts  
the corresponding result for finite-domain CSPs and chains of idempotent J\'{o}nsson operations (``\emph{congruence distributivity}'')~\cite{BartoKozikCD}, which was an intermediate step towards giving the full algebraic characterization of finite $\sA$ solvable by local-consistency methods (``\emph{congruence meet-semidistributivity}'')~\cite{BartoKozikBoundedWidth}. Our proof is, however, based on completely different techniques than the one for finite domains. This, the fact that the proof for quasi near-unanimity operations simply  mimics the one  for finite domains, and the fact that \mic{the above-mentioned result on canonical pseudo-Siggers operations is not purely algebraic due to the  additional condition of canonicity (which is a property that also requires a topology on the polymorphisms in order to be defined)},    
imply that 
here we provide the first non-trivial results showing that algebraic  conditions imply tractability for structures within Bodirsky-Pinsker conjecture.  By Propositions~6.9.9 and~6.9.12 in~\cite{BodirskyBook}, our main contribution also yields  tractability for templates invariant under \mic{chains of  quasi directed J\'{o}nsson operations and of quasi Pixley operations,  respectively}. 

\subsection{The main result}

Our general theorem firstly  assumes that the finite bounds of the ground structure $\relstr B$ in which the template is defined (see Conjecture~\ref{conj:BodirskyPinsker})  are closed under homomorphisms; we say that $\relstr B$ has  \emph{finite duality}. 
This is an important case within the scope of the Bodirsky-Pinsker conjecture,  including in particular many interesting  templates modeling Graph-SAT and  Hypergraph-SAT problems, 
 generalized completion problems, as well as the model-checking problems for MMSNP and GMSNP mentioned above.
Secondly, we assume that $\sB$ is \emph{$k$-neoliberal}, meaning that all its relations are $k$-ary and in a strong sense the entire structure of $\relstr B$ is completely reflected in these relations; a bit more precisely, there are no non-trivial first-order definable relations of arity smaller than $k$, all first-order definable relations of arity $k$ are unions of relations in $\sB$, all first-order definable  relations of higher arity are Boolean combinations of those of arity $k$, and  there are no algebraic dependencies between the elements of any $k$-tuple. 

The literature on CSPs  encompasses many examples of  $k$-neoliberal structures: for example,  Henson digraphs~\cite{BodPro25,feller2024algebraic,BitterM24}  (with $k=2$), the homogeneous C-relation (with $k=3$)~\cite{PhylogenyCSP},  and  $k$-uniform hypergraphs
for $k\geq 2$~\cite{Schaefer-Graphs,hypergraphs}. 
\mic{Among these, many moreover have finite duality (homogeneous graphs such as the random graph, Henson digraphs, and some $k$-uniform hypergraphs). These can} 
 be further generalized to 
 multi-graphs, i.e., graphs with edges colored 
 by several 
 colors,  to homogeneous multi-digraphs,  and to multi-hypergraphs, respectively. We would like to stress that our main result applies also to the latter structures  whose CSPs are far from being classified, and which present  significant obstacles to a classification. Indeed, while the literature on CSPs encloses a number of complexity classifications for first-order reducts of $k$-neoliberal structures $\relstr B$ whose automorphism group acts with two orbits on $k$-tuples with pairwise different entries (e.g., for $k$-uniform hypergraphs, $k$ pairwise different elements may either form a hyperedge or not),  
there are essentially no similar results when $\relstr B$ has hyperedges of $c$ pairwise different colors with $c > 2$ (a hyperedge and a non-hyperedge may be seen as two different colors).

\begin{example}
A general recipe for  obtaining finitely bounded homogeneous $3$-graphs (i.e.~graphs with three different kind of edges) is given in~\cite{cherlin_2022_2}. 
One starts with  any finite set of finite $2$-graphs 
$\mathcal{F}$ over a signature $\tau= \{ O, P \}$ where all pairs of vertices are connected either by an $O$-edge or a $P$-edge.
 The class of finite structures $\mathscr{C}$ over $\tau$ omitting the $2$-graphs in $\mathcal{F}$ has free amalgamation and the Fra\"{i}ss\'{e} limit $\sB$ of $\mathscr{C}$ yields  the desired  structure if we add to $O$ and $P$ a third kind of edge which connects two distinct vertices whenever  there is neither an $O$-edge nor a $P$-edge. It is a simple exercise to show that such 
$\sB$ is $2$-neoliberal and of finite duality.
\end{example}

When looking for examples of CSP templates with a first-order definition  in a  neoliberal structure with finite duality which are  preserved by quasi near-unanimity operations we reach out  to~\cite[Proposition~16]{Wrona:2020a}.

\begin{example}
Let $\sA$ be a structure whose all relations are first-order definable  as conjunctions of clauses of the form:
$$(x_1 \neq y_1 \vee \cdots \vee  x_k \neq y_1 \vee R(y_1, y_2) \vee y_2 \neq z_1 \vee \cdots \vee y_2 \neq z_l),$$
over the random graph $\sB =(V; E)$ where $R$ is either the edge relation $E$  or the non-edge relation $N$ imposed on pairwise different elements in $A$. Then $\sA$ is preserved by a quasi near-unanimity operation and in consequence by a chain of quasi J\'{o}nsson operations.
\end{example} 

An example of a relational structure preserved by a chain of quasi J\'{o}nsson operations but not by a quasi near-unanimity operation may be found among structures with a first-order definition in $(\mathbb{N}, =)$ where $\mathbb{N}$ is the set of natural numbers. These templates are usually called \emph{equality languages}  and have been investigated in detail  (up to primitive positive definability)  in~\cite{BodirskyCP10}. The desired equality language will be described by its polymorphisms.

\begin{example}
Let $1 \leq i \leq n$. An $n$-ary operation $f$ is \emph{injective in the $i$-th direction} if
$$x_i \neq \widehat{x}_i \implies f(x_1, \ldots, x_i, \ldots, x_n) \neq f(x_1, \ldots, \widehat{x}_i, \ldots, x_n)$$
for all $x_1, \ldots, x_i, \widehat{x}_i \ldots, x_n \in \mathbb{N}$. The operation $f$ is \emph{injective in one direction} if it is injective in the $i$-th direction for some $i \in \{ 1,\ldots, n \}$. We define $\mathbb{R}$ to be the structure containing all relations which are invariant under all operations that are injective in one direction. It is shown in~\cite{ScheckThesis} that $\mathbb{R}$ is invariant neither under a quasi near-unanimity operation nor under chains of idempotent Pixley or idempotent directed J\'{o}nsson operations; but it is invariant both under chains of quasi Pixley and quasi directed J\'{o}nsson operations, and hence also by chains of quasi J\'{o}nsson operations. 
\end{example}

\mic{Adopting the strategy for quasi near-unanimity operations from \cite{Wrona:2020b},}  
  we prove that any CSP template $\sA$ which is a first-order expansion of a ground structure $\relstr B$ with the properties described  above and which is invariant under a chain of quasi J\'{o}nsson operations has limited expressive power in the form of \emph{implicational simplicity}. We will show that this condition is equivalent to saying that $\sA$ does not primitively positively define an injective relation which entails a formula of the form $R(x_1,\dots,x_k)\Rightarrow R(x_{i_1},\dots,x_{i_k})$ for some relation $R$ primitively positively definable from $\sA$.

\begin{restatable}{theorem}{implsimple}\label{thm:implsimple}
Let $k\geq 2$, let $\sB$ be $k$-neoliberal, and suppose that $\sB$ has finite duality. Suppose that $\relstr A$ is a CSP template which is an expansion of $\sB$ by first-order definable relations (i.e., Boolean combinations of the relations of $\sB$). 
If $\relstr A$ is invariant under a chain of quasi J\'{o}nsson operations, then it is implicationally simple on injective instances.
\end{restatable}


As a corollary of~\Cref{thm:implsimple} and results from~\cite{SmoothApproximations, hypergraphs,MarimonPinskerMinimalOps}, we obtain a bound on the amount of local consistency needed to solve CSPs of the templates under consideration.

\begin{restatable}{corollary}{kmain}\label{thm:kmain}
Let $k\geq 2$, let $\sB$ be $k$-neoliberal, and suppose that $\sB$ has finite duality. Suppose that $\relstr A$ is a CSP template which is an expansion of $\sB$ by first-order definable relations (i.e., Boolean combinations of the relations of $\sB$). If $\relstr A$ is invariant under a chain of quasi J\'{o}nsson operations, then it has relational width $(k,\max(k+1,b_{\sB}))$, and hence $\Csp(\sA)$ is polynomial-time solvable.
\end{restatable}

\subsection{Related work}

\Cref{thm:kmain} provides both a tractability result and a 
\mic{bound on the amount of local consistency required to guarantee a solution to an instance.} \cht{Moreover, this bound is easily seen to be optimal for any structure under consideration.} 
\mic{As mentioned above, the only previous general results providing tractability within the Bodirsky-Pinsker conjecture by imposing conditions on polymorphisms were the one on quasi near-unanimity operations~\cite{BodirskyDalmau} (which is a straightforward generalization from the finite) and the one on canonical pseudo-Siggers operations~\cite{ReductsUnary} (which is not purely algebraic in the sense that it requires topology to be stated). 
We thus perceive~\Cref{thm:kmain} as the first non-trivial result on the tractability of infinite-domain CSPs originating from purely algebraic (height~1)  conditions.}

Building on results in~\cite{FederVardi,DalmauPearson,BartoKozikBoundedWidth}, Libor Barto has shown in~\cite{BartoCollapse} that the  relational width of any finite structure is either $(1,1)$ or $(2,3)$. We do not have a full understanding of this phenomenon over infinite structures. There is a plethora of interesting results, though.    
In~\cite{SmoothApproximations}, bounded width was characterized for CSP templates $\sA$ over several ``ground structures'' $\sB$ (in which they are first-order definable), using \emph{weak near-unanimity} identities satisfied by \emph{canonical polymorphisms}; this amounts to assuming weaker identities than quasi near-unanimity identities, but the additional (non-algebraic) property of canonicity.
The conditions given there were applied  in~\cite{SymmetriesEnough} to obtain a general  upper bound on the relational width of CSP templates satisfying them, and the bound was shown to be optimal for many templates. The bound on the relational width in~\Cref{thm:kmain} coincides with the bound proven in~\cite{SymmetriesEnough} for CSP templates which posses canonical \emph{pseudo-totally symmetric polymorphisms} of all arities $n\geq 3$. 
The results in~\cite{SymmetriesEnough} have been extended to CSP templates over finitely bounded homogeneous $k$-uniform hypergraphs for  $k\geq 3$ in~\cite{hypergraphs}.
\mic{The very first bounds as in}
~\Cref{thm:kmain} over infinite structures within Bodirsky-Pinsker conjecture were, however, obtained in~\cite{Wrona:2020a,Wrona:2020b}, where a similar upper bound was given for first-order expansions $\sA$ of certain binary structures $\sB$  under the assumption that $\sA$ is invariant under  a quasi near-unanimity operation.


\section{Preliminaries}\label{sect:Prelims}

\subsection{Relational structures and permutation groups}

Let $\sB$ be a relational structure, and let $\phi$ be a first-order formula over the signature of $\sB$. We identify the interpretation $\phi^{\sB}$ of $\phi$ in $\sB$ with the set of satisfying assignments for $\phi^{\sB}$. Let $V$ be the set of free variables of $\phi$, and let $\tuple u$ be a tuple of elements of $V$. We define $\proj_{\tuple u}(\phi^{\sB}):=\{f(\tuple u)\mid f\in \phi^{\sB}\}$. 
A \emph{first-order expansion} of a structure $\sB$ is an expansion of  $\sB$ by relations which are first-order definable in $\sB$, i.e., of the form $\phi^\sB$.

A first-order formula  is called  \emph{primitive positive} (pp) if it is built exclusively from  atomic formulae, existential quantifiers, and conjunctions. A  relation is \emph{pp-definable} in a relational structure  if it is first-order definable by a pp-formula. 

\begin{definition}\label{defn:fb}
Let $\sB$ be a structure over a finite relational signature $\tau$. We say that $\sB$ is \emph{finitely bounded} if there exists a finite set $\mathcal{F}$ of finite $\tau$-structures such that for every finite $\tau$-structure $\sC$, $\sC$ embeds to $\sB$ if no $\sF\in\mathcal{F}$ embeds to $\sC$.
Let $\mathcal{F}_{\sB}$ be a set witnessing the finite boudnedness of $\sB$ such that the size of the biggest structure contained in $\mathcal{F}_{\sB}$ is the smallest possible among all choices of $\mathcal F$; we write $b_{\sB}$ for this size.

We say that $\sB$ has a \emph{finite duality} if there exists a finite set $\mathcal{F}$ of finite $\tau$-structures such that for every finite $\tau$-structure $\sC$, $\sC$ maps homomorphically to $\sB$ if no $\sF\in\mathcal{F}$ maps homomorphically to $\sC$.

Finally, $\sB$ is \emph{homogeneous} if partial isomorphism between finite induced substructures extends to an automorphism of $\sB$. 
\end{definition}

Let $\gG$ be a permutation group acting on a set $A$, let $k\geq 1$,  and let $\tuple a\in A^k$. The \emph{orbit} of $\tuple a$ under $\gG$ is the set $\{g(\tuple a)\mid g\in \gG\}$. An orbit of a pair of elements is often called an \emph{orbital}. We say that $\gG$ is \emph{oligomorphic} if for every $k\geq 1$, $\gG$ has only finitely many orbits of $k$-tuples in its action on $A$. We say that a relational structure $\sA$ is \emph{$\omega$-categorical} if its automorphism group is oligomorphic.
Let $k\geq 1$. We say that $\fG$ is \emph{$k$-transitive} if it has only one orbit in its action on injective $k$-tuple of elements of $A$. 
$\fG$ is \emph{$k$-homogeneous} if for every $\ell\geq k$, the orbit of every $\ell$-tuple under $\fG$ is uniquely determined by the orbits of its $k$-subtuples. 
$\fG$ has \emph{no $k$-algebraicity} if the only fixed points of any stabilizer of $\fG$ by $k-1$ elements are these elements themselves. 
The \emph{canonical $k$-ary structure of $\fG$} is the relational structure on $A$ that has a relation for every orbit of $k$-tuples under $\fG$.

\begin{definition}
    Let $k\geq 2$, and let $\fG$ be a permutation group acting on a set $A$. 
    We say that $\fG$ is \emph{$k$-neoliberal} if it is oligomorphic, $(k-1)$-transitive, $k$-homogeneous, and has no $k$-algebraicity.

    A relational structure $\sB$ is \emph{$k$-neoliberal} if it is the canonical $k$-ary structure of a $k$-neoliberal permutation group.
\end{definition}

We remark that if a relational structure $\sB$ is $k$-neoliberal, then: by $(k-1)$-transitivity for every $\ell<k$, the only relations of arity $\ell$ which are first-order definable from $\sB$ are Boolean combinations of equalities and non-equalities; by definition, all $k$-ary first-order definable relations are unions of relations in $\sB$; and it follows from  $k$-homogeneity and oligomorphicity that any relation of arity $\ell> k$ first-order definable from $\sB$ is a  Boolean combination of the ($k$-ary) relations.

The notion of $k$-neoliberality is inspired by the notion of liberal binary cores from~\cite{Wrona:2020b} -- every liberal binary core is $2$-neoliberal. However, the opposite is not true -- the automorphism group of the universal homogeneous $\sK_3$-free graph (i.e., the unique homogeneous graph having the bounds $\mathcal F=\{\sK_3, \rightarrow, \circlearrowright \}$ in~\Cref{defn:fb}, where $\rightarrow$ is a directed edge and $\circlearrowright$ a loop) is easily seen to be $2$-neoliberal, whence its expansion by the equality relation and by the relation containing all pairs of distinct  elements which are not connected by an edge is $2$-neoliberal, but it is a binary core which is not liberal. This is because a liberal binary core is supposed to be finitely bounded and the set of forbidden bounds should not contain any structure of size $k$ whenever $3\leq k\leq 6$. However, $\sK_3$ is a $3$-element graph which does not embed into $\sB$ but all its subgraphs of size at most $2$ do, and hence $\sK_3$ has to be contained in any set of forbidden bounds for the universal homogeneous $\sK_3$-free graph.

\begin{example}
For every $k\geq 2$, the automorphism group of the universal homogeneous $k$-uniform hypergraph is $k$-neoliberal.

Let $\sC_\omega^2$ be the countably infinite equivalence relation where every equivalence class contains precisely $2$ elements. Then $\Aut(\sC_\omega^2)$ is oligomorphic, $1$-transitive, and $2$-homogeneous, but it is not $2$-neoliberal. Indeed, for any element $a$ of $\sC_\omega^2$, the stabilizer of $\Aut(\sC_\omega^2)$ by $a$ fixes also the unique element of $\sC_\omega^2$ which is in the same equivalence class as $a$.

On the other hand, the automorphism group of the countably infinite equivalence relation with equivalence classes of  fixed size $m>2$ is easily seen to be $2$-neoliberal.
\end{example}

Note that if $\gG$ is a permutation group acting on a set $A$ which is $k$-neoliberal for some $k\geq 2$ and which is not equal to the group of all permutations on $A$, then the number $k$ is uniquely determined. Indeed, $k=\min \{i\geq 1\mid \gG\text{ is not }i\text{-transitive}\}$.

\subsection{Constraint satisfaction problems and bounded width}

For $k\geq 1$, we write $[k]$ for the set $\{1,\ldots,k\}$. 
Let $k,\ell\geq 1$, let $A$ be a non-empty set, let $i_1,\ldots,i_{\ell}\in[k]$, and let $R\subseteq A^k$ be a relation. We write $\proj_{(i_1,\ldots,i_{\ell})}(R)$ for the $\ell$-ary relation $\{(a_{i_1},\ldots,a_{i_{\ell}})\mid (a_1,\ldots,a_k)\in R\}$.
For a tuple $\tuple t\in A^k$, we write $\scope(\tuple t)$ for its \emph{scope}, i.e., for the set of all entries of $\tuple t$.
We write $I_k^A$ for the relation containing all injective $k$-tuples of elements of $A$. 

Let $\sA$ be a relational structure. An \emph{instance of $\Csp(\sA)$} is a pair $\mathcal \instance=(\V,\constraints)$, where $\V$ is a finite set of variables and $\mathcal C$ is a finite set of \emph{constraints}; for every constraint $C\in \mathcal C$, there exists a non-empty set $U\subseteq \V$ called the \emph{scope} of $C$ such that $C\subseteq A^U$, and such that $C$ can be viewed as a relation of $\sA$ by totally ordering $U$; i.e., there exists an enumeration $u_1,\ldots,u_k$ of the elements of $U$ and a $k$-ary relation $R$ of $\sA$ such that for all $f\colon U\to A$, it holds that  
$f\in C$ if, and only if, $(f(u_1),\dots,f(u_k))\in R$. The relational structure $\sA$ is called the~\emph{template} of the CSP.
A~\emph{solution} of a CSP instance~$\instance = (\V,\constraints)$ is a mapping $f\colon \V\rightarrow A$ such that for every $C\in\constraints$ with scope $U$, $f|_U\in C$.

An instance $\instance=(\V,\constraints)$ of $\Csp(\sA)$ is \emph{non-trivial} if it does not contain any empty constraint; otherwise, it is \emph{trivial}.
Given a constraint $C\subseteq A^U$ with $U\subseteq \V$ and a tuple $\tuple{v}\in U^k$ for some $k\geq 1$, the \emph{projection of $C$ onto $\tuple{v}$} is defined by $\proj_{\tuple{v}}(C):=\{f(\tuple{v})\colon f\in C\}$.

We denote by $\Csp_{\injinstances}(\sA)$ the restriction of $\Csp(\sA)$ to those instances where for every constraint $C$ and for every pair of distinct variables $u,v$ in its scope, $\proj_{(u,v)}(C)\subseteq I_2^A$.

\begin{definition}\label{def:minimality}
Let $1\leq k\leq \ell$. We say that an instance $\instance=(\V,\constraints)$ of $\Csp(\sA)$ is \emph{$(k,\ell)$-minimal} if both of the following hold:
\begin{itemize}
\item the scope of every tuple of elements of $\V$ of length at most $\ell$ is contained in the scope of some constraint in $\constraints$;
\item for every $m\leq k$, for every tuple $\tuple u\in \V^m$, and for all constraints $C_1, C_2 \in \constraints$ whose scopes contain the scope of $\tuple u$, the projections of $C_1$ and $C_2$ onto $\tuple u$ coincide.
\end{itemize}
We say that an instance $\instance$ is \emph{$k$-minimal} if it is $(k,k)$-minimal.
\end{definition}

Let $1\leq k$. If $\instance=(\V,\constraints)$ is a $k$-minimal instance and $\tuple u$ is a tuple of variables of length at most $k$, then there exists a constraint in $\constraints$ whose scope contains the scope of $\tuple u$, and all the constraints who do have the same projection onto $\tuple u$.
We write $\proj_{\tuple u}(\instance)$ for this projection, and call it the \emph{projection of $\instance$ onto $\tuple u$}.

Let $1\leq k\leq \ell$, let $\sA$ be an $\omega$-categorical relational structure, and let $p$ denote the maximum of $\ell$ and the maximal arity of the relations of $\sA$. 
Clearly not every instance $\instance=(\V,\mathcal{C})$ of $\Csp(\sA)$ is $(k,\ell)$-minimal.
However, every instance $\instance$ is \emph{equivalent} to a $(k,\ell)$-minimal instance $\instance'$ of $\Csp(\sA')$ where $\sA'$ is the expansion of $\sA$ by all at most $p$-ary relations pp-definable in $\sA$ in the sense that $\instance$ and $\instance'$ have the same solution set.
In particular we have that if $\instance'$ is trivial, then $\instance$ has no solutions. Moreover, $\Csp(\sA')$ has the same complexity as $\Csp(\sA)$ and the instance $\instance'$ can be computed from $\instance$ in polynomial time (see e.g., Section~2.3 in~\cite{SymmetriesEnough} for the description of the algorithm).

\begin{definition}
Let $1\leq k\leq \ell$. 
A relational structure $\sA$ has \emph{relational width $(k,\ell)$} if every non-trivial $(k,\ell)$-minimal instance equivalent to an instance of $\Csp(\sA)$ has a solution. $\sA$ has \emph{bounded width} if it has relational width $(k,\ell)$ for some $k,\ell$.
\end{definition}

If $\sA$ has relational width $(k,\ell)$, then we will also say that $\Csp(\sA)$ has relational width $(k,\ell)$. We say that $\Csp_{\injinstances}(\sA)$ has relational width $(k,\ell)$ if every non-trivial $(k,\ell)$-minimal instance equivalent to an instance of $\Csp_{\injinstances}(\sA)$ has a solution.


\subsection{Polymorphisms}

Let $A$ be a set, let $k,n\geq 1$, and let $R\subseteq A^k$. A function $f\colon A^n\rightarrow A$ \emph{preserves} the relation $R$ if for all tuples $(a^1_1,\dots,a^1_k),\dots,(a^n_1,\dots,a^n_k)$
$\in R$, it holds that the tuple $(f(a^1_1,\dots,a^n_1),\dots,f(a^1_k,\dots,a^n_k))$ is contained in $R$ as well. 
The function is a \emph{polymorphism} of a relational structure $\sA$ if it preserves all relations of $\sA$. The set of all polymorphisms of $\sA$ is denoted by $\Pol(\sA)$. The importance of polymorphisms is based on the fact that for $\omega$-categorical $\sA$, the pp-definable relations are precisely those preserved by all polymorphisms of $\sA$~\cite{BN}.




The proof of the following useful observation and 
all other omitted proofs may be found in the appendix.

\begin{observation}
\label{obs:idemJonsson}
Let $\sA$ be an $\omega$-categorical model-complete core preserved by a chain of quasi J\'{o}nsson operations $(J_1, \ldots, J_{2n+1})$. Then for every finite subset $B \subseteq A$ there exists a chain of quasi J\'{o}nsson operations $(J'_1, \ldots, J'_{2n+1})$ such that each $J'_i$ with $i \in [2n+1]$ is idempotent on $B$, i.e. $J'_{i}(b, \ldots, b) = b$ for all $b \in B$.
\end{observation}

\noindent
The set of automorphisms of $\sA$ is denoted by $\Aut(\sA)$.
\section{Proof of the main result}

\subsection{Implicationally simple structures}

We introduce the notion of an implication and several related concepts  that will play a key role in the proof of~\Cref{thm:implsimple}. 

\begin{definition}\label{def:implication}
Let $\sA$ be a relational structure. Let $V$ be a set of variables, let $\tuple u,\tuple v$ be injective tuples of variables in $V$ of length $k<|V|$ and $m<|V|$, respectively, such that $\scope(\tuple u)\cup \scope(\tuple v)=V$. Let $C\subseteq A^k$ and $D\subseteq A^m$ be pp-definable from $\sA$ and non-empty. We say that a pp-formula $\phi$ over the signature of $\sA$ with free variables from $V$ is a \emph{$(C,\tuple u,D,\tuple v)$-implication in $\sA$} if all of the following hold:
\begin{enumerate}
    \item for  all \chmichal{$x \neq y$} in $V$, $\proj_{(x,y)}(\phi^{\sA})\not\subseteq\{(a,a)\mid a\in A\}$,
    \item $C\subsetneq\proj_{\tuple u}(\phi^{\sA})$,
    \item $D\subsetneq\proj_{\tuple v}(\phi^{\sA})$,
    \item for \chmichal{all} $f\in \phi^{\sA}$, it holds that $f(\tuple u)\in C$ implies $f(\tuple v)\in D$,
    \item for every $\tuple a\in D$, there exists $f\in\phi^{\sA}$ such that $f(\tuple u)\in C$ and $f(\tuple v)=\tuple a$.
\end{enumerate}

We say that $\phi$ is a \emph{$(C,\tuple u,D,\tuple v)$-pre-implication} if it satisfies items (2)-(5).
We will call $\phi$ a \emph{$(C,D)$-implication} if it is a $(C,\tuple u,D,\tuple v)$-implication for some $\tuple u\in I^V_k, \tuple v\in I^V_m$. We say that an implication $\phi$ is \emph{injective} if $\phi^{\sA}$ contains only injective mappings.

Let $\gG$ be a permutation group acting on $A$, and let $f\in\phi^{\sA}$. If $O,P$ are orbits under $\fG$ such that $f(\tuple u)\in O$, $f(\tuple v)\in P$, then we say that $f$ is an \emph{$OP$-mapping}.
\end{definition}

\begin{example}
Let $\sA$ be a relational structure, let $k\geq 1$, and let $O$ be an orbit of $k$-tuples under $\Aut(\sA)$. Suppose that $\sA$ pp-defines the equivalence of orbits of $k$-tuples under $\Aut(\sA)$. Then the formula defining this equivalence is an $(O,O)$-pre-implication in $\sA$. If $\sA$ is such that $\Aut(\sA)$ does not have any fixed point in its action on $A$, this pre-implication is an implication. For all orbits $P,Q$ of $k$-tuples under $\Aut(\sA)$, $\phi^{\sA}$ contains an $PQ$-mapping if, and only if, $P=Q$.
\end{example}

\begin{definition}
Let $\sA$ be a relational structure, and let $k\geq 1$.

The \emph{$k$-ary implication graph of $\sA$}, to be denoted by $\mathcal{G}_{\sA}$, is a directed graph defined as follows.

\begin{itemize}
    \item The set of vertices is the set of pairs $(C_1,C)$ where $\emptyset\neq C\subsetneq C_1\subseteq A^k$ and $C,C_1$ are pp-definable from $\sA$.
    \item There is an arc from $(C_1,C)$ to $(D_1,D)$ if there exists a $(C,\tuple u,D,\tuple v)$-implication $\phi$ in $\sA$ such that $\proj_{\tuple u}(\phi^{\sA})=C_1$, $\proj_{\tuple v}(\phi^{\sA})=D_1$.
\end{itemize}

The \emph{$k$-ary injective implication graph of $\sA$}, denoted by $\mathcal{G}_{\sA}^{\injinstances}$, is the (non-induced) subgraph of $\mathcal{G}_{\sA}$ 
that contains precisely the vertices $(C_1,C)$ where $C_1$ is injective and which contains an arc from $(C_1,C)$ to $(D_1,D)$ if $(C_1,C)\neq (D_1,D)$ and if there exists an injective $(C,\tuple u,D,\tuple v)$-implication $\phi$ in $\sA$ with $\proj_{\tuple u}(\phi^{\sA})=C_1$, $\proj_{\tuple v}(\phi^{\sA})=D_1$.

We say that $\sA$ is \emph{implicationally simple (on injective instances)} if the (injective) implication graph $\mathcal{G}_{\sA}$ ($\mathcal{G}_{\sA}^{\injinstances}$) is acyclic. Otherwise, $\sA$ is \emph{implicationally hard (on injective instances)}.
\end{definition}

Note that by item (1) in~\Cref{def:implication}, the implication graph does not necessarily contain all loops -- e.g., the formula over variables $\{x_1,\ldots,x_{2k}\}$ defined by $\bigwedge\limits_{i\in[k]}(x_i=x_{i+k})$ is not an implication.

The following is essentially subsumed by Lemma 3.3 in~\cite{SymmetriesEnough} but we provide the reformulation to our setting for the convenience of the reader.

\begin{lemma}\label{lemma:core-bwidth}
Let $k\geq 2$, let  $\sB$ be $k$-neoliberal, and suppose that $\sB$ is finitely bounded. Let $\instance=(\V,\constraints)$ be a non-trivial, $(k,\max(k+1,b_{\sB}))$-minimal instance of $\Csp(\sB)$ such that for every $\tuple v\in\V^k$, $\proj_{\tuple v}(\instance)$ contains precisely one orbit under $\fG$. Then $\instance$ has a solution.
\end{lemma}

\begin{proposition}\label{prop:implsimple}
Let $k\geq 2$, let $\sB$ be $k$-neoliberal, and suppose that $\sB$ has finite duality. Let $\sA$ be a first-order expansion of $\sB$ which is implicationally simple on injective instances. 
Then $\sA$ has relational width $(k,\max(k+1,b_{\sB}))$.
\end{proposition}

\subsection{Binary injections and quasi J\'{o}nsson operations}

Here, we restate some results about binary operations from~\cite{SmoothApproximations,MarimonPinskerMinimalOps} and prove a new result that will enable us to use Lemma 14 from~\cite{hypergraphs} in order to prove~\Cref{cor:libcores-purelyinj}.



\begin{lemma}
\label{lemma:binessen}[Corollary 3.7 in~\cite{MarimonPinskerMinimalOps}]
Let $\sB$ be an $\omega$-categorical countable model-complete core such that $\Aut(\sB)$ has $\leq 2$
orbits. Then, if $\Pol(\sB)$ has an essential polymorphism, it also has a binary essential polymorphism.
\end{lemma}

\begin{lemma}[Proposition 24 in~\cite{SmoothApproximations}]\label{lemma:bininj2}
Let $\sA$ be a first-order reduct of an $\omega$-categorical structure $\sB$ such that $\Aut(\sB)$ is $1$-transitive and such that its canonical binary structure has finite duality. If $\Pol(\sA)$ contains a binary essential function preserving $I_2^B$, then it contains a binary injection.
\end{lemma}

\begin{proposition}\label{prop:libcores-purelyinj}
Let $\ell\geq k\geq 2$, let $\sA$ be an $\omega$-categorical relational structure, and suppose that $I_2^A$ is pp-definable from $\sA$ and that $\Pol(\sA)$ contains a binary injection. Then $\sA$ has relational width $(k,\ell)$ if, and only if, $\Csp_{\injinstances}(\sA)$ has relational width $(k,\ell)$.
\end{proposition}

\Cref{lemma:binessen,lemma:bininj2,prop:libcores-purelyinj} immediately yield the following corollary, which will enable us to reduce $\Csp(\sA)$ for any structure $\sA$ in the scope of~\Cref{thm:implsimple} to $\Csp_{\injinstances}(\sA)$. The reason for $k\geq 3$ is that $I_2^B$ is pp-definable from $\sB$ by the $2$-transitivity  of $\Aut(\sB)$. 
On the other hand, for $k=2$, the structures  considered in~\Cref{thm:implsimple} are only $1$-transitive and do not automatically 
pp-define $I_2^B$. We will, however, show in the following subsections that they do pp-define $I_2^B$ under the assumption that they are preserved by a chain of quasi J\'{o}nsson operations, see~\Cref{lemma:diseqJonsson}.

\begin{corollary}\label{cor:libcores-purelyinj}
Let $k\geq 2$, let 
$\sB$ be $k$-neoliberal
,
of finite duality and such that it pp-defines $I^B_2$. Let $\sA$ be a first-order expansion of $\sB$, and suppose that $\sA$ is preserved by a chain of quasi J\'{o}nsson operations. 
Then $\sA$ has relational width $(k,\max(k+1,b_{\sB}))$ if, and only if, $\Csp_{\injinstances}(\sA)$ has relational width $(k,\max(k+1,b_{\sB}))$.
\end{corollary}


\subsection{Some implications which are not preserved by quasi J\'{o}nsson operations}

In this section, we first prove that if a structure $\sA$ pp-defines certain implications, then it is not preserved by any chain of quasi J\'{o}nsson operations 
(\Cref{lemma:l+1=rel,lemma:l+2=rel}). This will enable us to prove that if $\sA$ 
is preserved by a chain of quasi J\'{o}nsson operations, and if a relation pp-definable in $\sA$ contains a tuple with certain properties, then this relation contains an injective tuple with the same properties (\Cref{cor:noequality}).

\begin{lemma}\label{lemma:l+1=rel}
Let $k\geq 2$, and let $\sB$ be $k$-neoliberal, and suppose that $\sB$ has finite duality. Let $\sA$ be a first-order expansion of $\sB$, let $\ell\in\{2,\ldots,k\}$, let $T\subseteq I_{\ell}^B$, and let $\phi$ be a $(T,=)$-implication in $\sA$ with $\ell+1$ variables. Then $\sA$ is not preserved by any chain of quasi J\'{o}nsson operations.
\end{lemma}

\begin{proof}
Enumerate the variables of $\phi$ by $x_1,\ldots,x_{\ell+1}$. Without loss of generality, $\tuple u=(x_1,\ldots,x_{\ell})$ and $\tuple v=(x_{\ell},x_{\ell+1})$ are such that $\phi$ is an $(T,\tuple u,=,\tuple v)$-implication in $\sA$. The set $\phi^{\sA}$ can then be viewed as an $(\ell+1)$-ary relation $R(x_1,\ldots,x_{\ell+1})$.

Using the $k$-neoliberality and finite duality of $\sB$, we can find $n\geq 1$ and $S := \{a^i_j\in B\mid 1\leq i\leq n-1, 1\leq j\leq \ell\}
\cup \{b^i_j\in B\mid 1\leq i\leq n-1, 1\leq j\leq \ell\} \cup \{ a^n_{\ell}, b^n_{\ell} \} \cup \{ d_1, \ldots, d_{\ell-1} \}
$ 
such that all of the following hold:
\begin{itemize}
    \item $a_{\ell}^1 = b_{\ell}^n,a_{\ell}^n = b_{\ell}^1$,
    \item $a^i_\ell\neq a^{i+1}_\ell$, $b^i_\ell\neq b^{i+1}_\ell$ for all $i\in[n-1]$,
    \item $(d_1,\ldots, d_{\ell-1}, a^1_\ell )\in T$,
    \item $(d_1,\ldots, d_{\ell-1}, a^n_\ell)\in T$,
    \item $(a^i_1,\ldots,a^i_{\ell-1}, a^i_\ell, a^{i+1}_\ell) \in R$ for all $i \in [n-1]$,
    \item $(b^i_1,\ldots, b^i_{\ell-1},b^i_\ell, b^{i+1}_\ell) \in R$ for all $i \in [n-1]$.
\end{itemize}

To see this, let $(d_1,\ldots,d_{\ell -1}, a^1_{\ell})\in T$ be arbitrary.
The fact that the automorphism group of $\sB$ has no $k$-algebraicity implies that there exists $a^n_{\ell}\in B$, the choice of the value of $n$ is made later,  which is distinct from $a^1_{\ell}$ but which lies in the same orbit under the stabilizer of $\Aut(\sB)$ by $a^1_1,\ldots, a^1_{\ell-1}$. In particular, $(d_1,\ldots,d_{\ell - 1}, a^1_{\ell})$ and $(d_1,\ldots, d_{\ell - 1}, a^n_{\ell})$ lie in the same orbit under $\Aut(\sB)$, and hence $(d_1,\ldots, d_{\ell - 1}, a^n_{\ell}) \in T$. 
If $k\geq 3$, we set $n:=2$, $b^n_\ell:=a^1_\ell$, and $b^1_\ell:=a^n_\ell$. Since $\Aut(\sB)$ is $2$-transitive and $\proj_{(\ell,\ell+1)}(R)\not\subseteq \{(a,a)\mid a\in A\}$, 
it follows that $I_2^B\subseteq\proj_{(\ell,\ell+1)}(R)$, and hence we can find $a^1_1,\ldots,a^1_{\ell-1}$ and 
$b^1_1,\ldots,b^1_{\ell-1}$
such that $(a^1_1,\ldots,a^1_{\ell-1},a^1_\ell,a^n_{\ell})\in R$ and
$(b^1_1,\ldots,b^1_{\ell-1},b^1_\ell,b^n_{\ell}) \in R$.

If $k=2$, then $\ell=2$. In this case, 
we take $n$ such that for every finite structure $\sX$ in the signature $\tau$ of $\sB$, it holds that if every substructure of $\sX$ of size at most $n-1$ maps homomorphically to $\sB$, then so does $\sX$. This is possible by the finite duality of $\sB$. Let us define a $\tau$ structure on a set $\{x^i\mid 1\leq i\leq n\} \cup \{ y^i\mid 1 \leq i\leq n\}$ with $x^1 = y^n$ and $x^n = y^1$ as follows. Recall that the relations from $\tau$ correspond to the orbits of injective pairs under $\Aut(\sB)$. Let $O$ be the orbit of $(a^1_2,a^n_2)$, and let $P$ be an arbitrary injective orbit contained in $\proj_{(\ell,\ell+1)}(R)$. Such an orbit exists since $\proj_{(\ell,\ell+1)}(R)\not\subseteq \{(a,a)\mid a\in A\}$. We first define $R^{\sX}$ to be an empty relation for every $R\in \tau$, and we add the tuple $(x^1,x^n)$ to $O^{\sX}$ and all the tuples $(x^i,x^{i+1})$  as well as $(y^i,y^{i+1})$ for $i\in[n-1]$
to $P^{\sX}$. Observe that every substructure of $\sX$ of size at most $n-1$ has a homomorphism to $\sB$ by the $1$-transitivity of $\Aut(\sB)$. It follows by the choice of $n$ that $\sX$ has a homomorphism $h$ to $\sB$. Moreover, by the $2$-homogeneity of $\sB$, we can assume that $h(x^1)=a^1_2, \ldots, h(x^n)=a^n_2$ and 
$h(y^1)=b^1_2, \ldots, h(y^n)=b^n_2$. Thus, we set $a^i_2:=h(x^i)$ and  
 $b^i_2:=h(y^i)$
for every $i\in[n]$. It follows that for every $i\in[n-1]$, $(c^i_2,c^{i+1}_2)\in\proj_{(2,3)}(R)$ for $c \in \{a,b\}$, whence we can find $c^i_1$ such that $(c^i_1,c^i_2,c^{i+1}_2)\in R$ as desired.

Suppose for contradiction that the structure $\sA$ is preserved by a chain $(J_1, \ldots, J_{2m+1})$ of quasi J\'{o}nsson operations. Then by~\Cref{obs:idemJonsson}, we may assume without loss of generality that all the operations
$J_1, \ldots, J_{2m+1}$ are idempotent on $S$. 
Striving for contradiction, we first show in~\Cref{claim:equall+1} 
that $J_{2m+1}(a^1_\ell, a^n_\ell, a^n_\ell) = a^1_\ell$. Afterwards, we show that
$J_{2m+1}(a^n_\ell, a^n_\ell, a^n_\ell) = a^1_\ell$
which contradicts the idempotency of $J_{2m+1}$.

\begin{claim}
\label{claim:equall+1}
For all $k \in [2m+1]$ we have both of the following.
\begin{itemize}
\item $J_k(a^1_\ell, a^1_\ell, a^n_\ell) = a^1_\ell$ and
\item $J_k(a^1_\ell, a^n_\ell, a^n_\ell) = a^1_\ell$.
\end{itemize}
\end{claim}

\begin{proof}
The proof goes by induction on $k \in [2m+1]$. 

\textbf{(Base case: $k = 1$)}
For the first item we simply use the idempotency of $J_1$ on $S$ and~(\ref{eq:J1}) in Definition~\ref{def:QuasiJonsson}. In order to show the second item, we prove by  induction on $i \in [n]$ that $J_1(a^1_\ell, a^i_\ell, a^n_\ell) = a^1_\ell$.
By the first item, we have that $J_1(a^1_\ell, a^1_\ell, a^n_\ell) = a^1_\ell$.  For the induction step, we assume 
 $J_1(a^1_\ell,a^{i}_\ell, a^{n}_\ell) = a^1_{\ell}$
for some $i\in[n- 1]$, and show
  $J_1(a^1_\ell,a^{i+1}_\ell, a^{n}_\ell) = a^1_{\ell}$.

By the discussion above and since $\phi$ is a $(T,=)$-implication, we have that 
\begin{align*}
   \tuple t_{1,1} := \begin{pmatrix}
       d_1 \\
       \ldots \\
       d_{\ell-1} \\
       a^1_\ell \\
       a^1_\ell
    \end{pmatrix},
    \tuple t_{i,i+1} := \begin{pmatrix}
       a^{i}_1 \\
       \ldots \\
       a^{i}_{\ell-1} \\
       a^i_\ell \\
       a^{i+1}_\ell
    \end{pmatrix},
      \tuple t_{n,n} := \begin{pmatrix}
       d_1 \\
       \ldots \\
       d_{\ell-1} \\
       a^n_\ell \\
       a^n_\ell
    \end{pmatrix}.
\end{align*}
 for all $i \in [n-1]$ are in $R$. 
Since $f$ preserves $R$, it follows that $$f(\tuple t_{1,1}, \tuple t_{i,i+1},\tuple t_{n,n})\in R,$$ for all $i \in [n-1]$. This means that
\begin{align*}
    \begin{pmatrix}
       d_1 &=& J_1(d_1,a^i_1,d_1) \\
       \ldots && \ldots \\
       d_{\ell-1} &=& J_1(d_{\ell-1},a^i_{\ell-1},d_{\ell-1}) \\
       a^1_{\ell} &=& J_1(a^1_\ell,a^i_\ell,a^n_\ell) \\
    && J_1(a^1_\ell,a^{i+1}_\ell,a^n_\ell)
    \end{pmatrix} \in R.
\end{align*}
where the first $\ell-1$ equalities follow by~(\ref{eq:Ji}) in Definiton~\ref{def:QuasiJonsson} and the last of the displayed disequalities by the induction hypothesis. 
Since $(d_{1},\ldots, d_{\ell - 1}, a^1_\ell)\in T$ and $\phi$ is an $(T,=)$-implication, we obtain $J_1(a^1_\ell,a^i_\ell, a^n_\ell) =  J_1(a^1_\ell,a^{i+1}_\ell, a^n_\ell) = a^{1}_{\ell}$, as desired.

\noindent
\textbf{(Induction step.)}
We now assume that 
both items in the formulation of the claim hold for some $k \in [2m]$ and prove that both items hold for $k+1$. The proof depends on whether $k$ is odd or even. In the latter case we have $J_{k+1}(a^1_\ell, a^1_\ell, a^n_\ell) = J_{k}(a^1_\ell, a^1_\ell, a^n_1) = a^1_1$ by~(\ref{eq:J2i2i+1}). To obtain 
$J_{k+1}(a^1_\ell, a^n_\ell, a^n_\ell) = a^1_\ell$ we proceed in the same way as in the base case. 

On the other hand, if $k$ is odd, then by~(\ref{eq:J2i-12i}) we have $J_{k+1}(a^1_{\ell}, a^n_{\ell}, a^n_{\ell}) = J_{k}(a^1_{\ell}, a^n_{\ell}, a^n_{\ell}) = a^1_{\ell}$. Since $a^n_{\ell} = b^1_{\ell}$, it means that $J_{k+1}(a^1_{\ell}, b^1_{\ell}, a^n_{\ell}) = a^1_{\ell}$. 
By 
induction on $i \in [n]$, we will now show that $J_{k+1}(a^1_{\ell}, b^i_{\ell}, a^n_{\ell}) = a^1_{\ell}$. Since $b^n_\ell$ is $a^1_\ell$
we will have $J_{k+1}(a^1_{\ell}, a^1_{\ell}, a^n_{\ell}) = a^1_{\ell}$ as desired.
The base case is demonstrated. For the induction step 
assume that $J_{k+1}(a^1_{\ell}, b^i_{\ell}, a^n_{\ell}) = a^1_{\ell}$ for some $i$. We will show that $J_{k+1}(a^1_{\ell}, b^{i+1}_{\ell}, a^n_{\ell}) = a^1_{\ell}$. 
To this end, we observe that 
\begin{align*}
   \tuple s_{1,1} := \begin{pmatrix}
       d_1 \\
       \ldots \\
       d_{\ell-1} \\
       a^1_\ell \\
       a^1_\ell
    \end{pmatrix},
    \tuple s_{i,i+1} := \begin{pmatrix}
       b^i_1 \\
       \ldots \\
       b^i_{\ell-1} \\
       b^i_\ell \\
       b^{i+1}_\ell
    \end{pmatrix},
      \tuple s_{n,n} := \begin{pmatrix}
       d_1 \\
       \ldots \\
       d_{\ell-1} \\
       a^n_\ell \\
       a^n_\ell
    \end{pmatrix}.
\end{align*}
are 
tuples in $R$ by construction. 
Since $J_{k+1}$ preserves $R$, it follows that $$J_{k+1}(\tuple s_{1,1}, \tuple s_{i, i+1},\tuple s_{n,n}) \in R,$$ for all $i \in [n-1]$. This means that
\begin{align*}
    \begin{pmatrix}
       d_1&=& J_{k+1}(d_1,b^i_1,d_1) \\
       \ldots && \ldots \\
       d_{\ell-1} &=& J_{k+1}(d_{\ell-1},b^i_{\ell-1},d_{\ell-1}) \\
       a^1_{\ell} &=& J_{k+1}(a^1_\ell,b^i_\ell,a^n_\ell) \\
    &&J_{k+1}(a^1_\ell,b^{i+1}_\ell,a^n_\ell)
    \end{pmatrix} \in R.
\end{align*}
where the first $\ell-1$ equalities follow by~(\ref{eq:Ji}) and the last but one by the induction hypothesis. 
Since $(d_1,\ldots,d_{\ell-1},a^1_{\ell}) \in T$ and $\phi$ is an $(T,=)$-implication, we obtain $J_{k+1}(a^1_\ell,b^i_\ell, a^n_\ell) =  J_{k+1}(a^1_\ell,b^{i+1}_\ell, a^n_\ell) = a^{1}_{\ell}$, as desired. 
It completes the proof of the claim.
\end{proof}

By the claim, we have $J_{2m+1}(a^1_\ell, a^n_\ell, a^n_\ell) = a^1_\ell$. In order to complete the proof of the lemma, we will show by induction
on $i \in [n]$ that $J_{2m+1}(a^i_\ell, a^n_\ell, a^n_\ell) = a^1_\ell$ for all $i \in [n]$. The base case is already done. 
For the induction step assume that  we have $J_{2m+1}(a^i_\ell, a^n_\ell, a^n_\ell) = a^1_\ell$ for some $i \in [n]$.

Since $J_{2m+1}$ preserves $R$, it follows that $$J_{2m+1}(\tuple t_{i,i+1}, \tuple t_{n,n},\tuple t_{n,n})\in R,$$ for all $i \in [1, n-1]$. This means that
\begin{align*}
    \begin{pmatrix}
       d_1 &=& J_{2m+1}(a^i_1,d_1,d_1) \\
       \ldots && \ldots \\
       d_{\ell-1} &=& J_{2m+1}(a^i_{\ell-1},d_{\ell-1},d_{\ell-1}) \\
       a^1_{\ell} &=& J_{2m+1}(a^i_\ell,a^n_\ell,a^n_\ell) \\
    &&J_{2m+1}(a^{i+1}_\ell,a^n_\ell,a^n_\ell)
    \end{pmatrix} \in R.
\end{align*}
where the first $\ell-1$ equalities follow by~(\ref{eq:J2n+1}) and the last one by the induction hypothesis. 
Since $(d_1,\ldots,d_{\ell-1},a^1_\ell)\in T$ and $\phi$ is a $(T,=)$-implication, we obtain $J_{2m+1}(a^i_\ell,a^{n}_\ell, a^n_\ell) =  J_{2m+1}(a^{i+1}_\ell,a^{n}_\ell, a^n_\ell) = a^{1}_{\ell}$, as desired. 
It completes the proof of the lemma.
\end{proof}

\begin{lemma}\label{lemma:l+2=rel}
Let $k\geq 2$, an let $\sB$ be $k$-neoliberal, and suppose that $\sB$ has finite duality. Let $\sA$ be a first-order expansion of $\sB$, let $\ell\in\{2,\ldots,k\}$, let $T\subseteq I_{\ell}^B$, and let $\phi$ be an $(T,=)$-implication in $\sA$ with $\ell+2$ variables. Then $\sA$ is not preserved by any chain of quasi J\'{o}nsson operations.
\end{lemma}

The following corollary follows from~\Cref{lemma:l+1=rel,lemma:l+2=rel}.

\begin{corollary}\label{cor:noequality}
Let $k\geq 2$, let $\sA$ be a first-order expansion of a $k$-neoliberal $\sB$ with finite duality and suppose that $\sA$ is preserved by a chain of quasi J\'{o}nsson operations. Let $\phi$ be a pp-formula over the signature of $\sA$ with variables from a set $V$ such that for all distinct $x,y\in V$, $\proj_{(x,y)}(\phi^{\sA})\not\subseteq\{(a,a)\mid a\in A\}$, and let $g\in\phi^{\sA}$. Then there exists an injective $h\in\phi^{\sA}$ with the property that for every $r\geq 1$ and for every $\tuple v\in I^V_r$, if $g(\tuple v)$ is injective, then $g(\tuple v)$ and $h(\tuple v)$ belong to the same orbit under $\Aut(\sB)$.
\end{corollary}

\subsection{Composition of implications}
We introduce composition of implications which will play an important role in the rest of the article.

\begin{definition}\label{def:composition}
Let $\sA$ be a relational structure, let $k\geq 1$, let $C,D,E\subseteq A^{k}$ be non-empty, let $\phi_1$ be a $(C,\tuple u^1,D,\tuple v^1)$-implication in $\sA$, and let $\phi_2$ be a $(D,\tuple u^2,E,\tuple v^2)$-implication in $\sA$. Let us rename the variables of $\phi_2$ so that $\tuple v^1=\tuple u^2$ and so that $\phi_1$ and $\phi_2$ do not share any other variables.
We define $\phi_1\circ\phi_2$ to be the pp-formula arising from the formula $\phi_1\wedge \phi_2$ by existentially quantifying all variables that are not contained in $\scope(\tuple u^1)\cup \scope(\tuple v^2)$.

Let $\psi$ be a $(C,\tuple u,C,\tuple v)$-implication. For $n\geq 2$, we write $\psi^{\circ n}$ for the pp-formula $\psi\circ \cdots \circ \psi$ where $\psi$ appears exactly $n$ times.
\end{definition}

\begin{lemma}\label{lemma:kimpl}
Let $k\geq 2$, let $\sA$ be a first-order expansion of the canonical $k$-ary structure $\sB$ of a permutation group $\fG$. 
Let $\phi_1,\phi_2$ be as in~\Cref{def:composition}, and suppose that $\proj_{\tuple v^1}(\phi_1)= \proj_{\tuple u^2}(\phi_2)$.
Then $\phi:=\phi_1\circ \phi_2$ is a $(C,\tuple u^1,E,\tuple v^2)$-pre-implication in $\sA$. Moreover, for all orbits $O_1\subseteq \proj_{\tuple u^1}(\phi_1^{\sA})$, $O_3\subseteq \proj_{\tuple v^2}(\phi_2^{\sA})$ under $\gG$, $\phi^{\sA}$ contains an $O_1 O_3$-mapping if, and only if, there exists an orbit $O_2$ under $\gG$ such that $\phi_1^{\sA}$ contains an $O_1 O_2$-mapping and $\phi_2^{\sA}$ contains an $O_2 O_3$-mapping.

Suppose moreover that $\gG$ is $k$-neoliberal, and that $\sB$ has finite duality and that $\sA$ is preserved by a chain of quasi J\'{o}nsson operations. Then there exists an injective $O_1O_3$-mapping in $\phi^{\sA}$ provided 
$\phi_1^{\sA}$ contains an injective $O_1 O_2$-mapping, and $\phi_2^{\sA}$ contains an injective $O_2 O_3$-mapping.
 
\end{lemma}

\noindent
Here comes a straightforward corollary of the  lemma above.
\begin{corollary}\label{cor:kimpl}
Let $\sA$, $\fG$, $\phi_1,\phi_2, \phi$ be as in~\Cref{lemma:kimpl}.
If $\phi_1$ and $\phi_2$ are injective implications. Then $\phi$ is a $(C,\tuple u^1,E,\tuple v^2)$-implication in $\sA$. Restricting $\phi^{\sA}$ to injective mappings, one moreover obtains an injective $(C,\tuple u^1,E,\tuple v^2)$-implication, which for all injective orbits $O_1\subseteq C, O_3\subseteq E$ under $\fG$ contains an $O_1 O_3$-mapping if, and only if, $\phi^{\sA}$ contains such mapping.
\end{corollary}

The following observation states a few properties of implications and their compositions which will be used later.

\begin{observation}\label{observation:impl_properties}
    Let $\sA$ be a relational structure, let $k\geq 2$, let $C\subseteq A^k$, let $\phi_1$ be a $(C,\tuple u^1,C,\tuple v^1)$-implication in $\sA$, and let $\phi_2$ be a $(C,\tuple u^2,C,\tuple v^2)$-implication in $\sA$. Let $p_1$ be the number of variables of $\phi_1$, let $p_2$ be the number of variables of $\phi_2$, and let $p$ be the number of variables of $\phi:=\phi_1\circ \phi_2$. Then all of the following hold.

    \begin{enumerate}
        \item $p\geq\max(p_1,p_2)$.
        \item $p=p_1=p_2$ if, and only if, $\scope(\tuple u^1)\cap\scope(\tuple v^2) = \scope(\tuple u^1)\cap\scope(\tuple v^1)=\scope(\tuple u^1)\cap\scope(\tuple u^2)\cap\scope(\tuple v^2)$.
        \item Suppose that $\phi_1=\phi_2$ and $p=p_1$. Then for every $i\in[k]$, it holds that if $v^1_i$ is contained in the intersection in (2), then so is $u^1_i$.
    \end{enumerate}
\end{observation}

\subsection{Digraphs of implications}

We reformulate the notion of digraph of implications from~\cite{Wrona:2020b} and prove a few auxiliary statements about these digraphs.

\begin{definition}\label{def:impl}
Let $k\geq 2$, 
let $\sA$ be a first-order expansion of a $k$-neoliberal relational structure $\sB$. Let $\emptyset\neq C\subseteq A^k$, and let $\phi$ be $(C,\tuple u,C,\tuple v)$-implication in $\sA$ such that $\proj_{\tuple u}(\phi^{\sA})=\proj_{\tuple v}(\phi^{\sA})=:E$. Let $\mathbf{Vert}(E)$ be the set of all orbits under $\Aut(\sB)$ contained in $E$.
Let $\mathcal{B}_{\phi}\subseteq \mathbf{Vert}(E)\times \mathbf{Vert}(E)$ be the directed graph such that $\mathcal{B}_{\phi}$ contains an arc $(O,P)\in \mathbf{Vert}(E)\times \mathbf{Vert}(E)$ if $\phi^{\sA}$ contains an $OP$-mapping.

We say that $S \subseteq \mathbf{Vert}(E)$ is a \emph{strongly connected component} if it is a maximal set with respect to inclusion such that for all (not necessary distinct) vertices $O,P\in S$, there exists a path in $\mathcal{B}_{\phi}$ connecting $O$ and $P$. We say that $S$ is a \emph{sink} in $\mathcal B_{\phi}$ if every arc originating in $S$ ends in $S$; $S$ is a \emph{source} in $\mathcal B_{\phi}$ if every arc finishing in $S$ originates in $S$.
\end{definition}

Note that the digraph $\mathcal B_\phi$ can be defined also for relational structures which do not satisfy the assumptions on $\sA$ from~\Cref{def:impl}; however, these assumptions are needed in the proof of~\Cref{lemma:kcomplete} so we chose to include them already in~\Cref{def:impl}. Note also that $\mathcal{B}_\phi$ can contain vertices which are not contained in any strongly connected component.

\begin{observation}\label{observation:sinksource}
Let $\phi$ be as in~\Cref{def:impl}.  
Then there exist strongly connected components $S_1\subseteq\mathbf{Vert}(C),S_2\subseteq\mathbf{Vert}(E\backslash C)$ in $\mathcal B_\phi$ such that $S_1$ is a sink in $\mathcal{B}_{\phi}$, and $S_2$ is a source in $\mathcal{B}_{\phi}$.
Moreover, since $\proj_{\tuple u}(\phi^{\sA})=\proj_{\tuple v}(\phi^{\sA})=E$, any vertex $O\in \mathbf{Vert}(E)$ has an outgoing and an incoming arc, i.e., the digraph $\mathcal{B}_{\phi}$ is \emph{smooth}.
\end{observation}

Let $\phi$ be as in~\Cref{def:impl}, and set $I_\phi:=\{i\in[k]\mid u_i\in\scope(\tuple v)\}$. If the number of variables of $\phi$ is equal to the number of variables of $\phi\circ \phi$, then item (3) in~\Cref{observation:impl_properties} yields that $I_\phi=\{i\in[k]\mid v_i\in\scope(\tuple u)\}$. 

\begin{definition}\label{def:kcomplimpl}
Let $\phi$ be as in~\Cref{def:impl}. We say that $\phi$ is \emph{complete} if the number of variables of $\phi\circ \phi$ is equal to the number of variables of $\phi$, $u_i=v_i$ for every $i\in I_\phi$, and each strongly connected component of $\mathcal{B}_{\phi}$ contains all possible arcs including loops.
\end{definition}

The following is a modification of Lemma 36 in \cite{Wrona:2020b}:

\begin{lemma}\label{lemma:kcomplete}
Let $\phi$ be an injective implication as in~\Cref{def:impl}. If $\sB$ pp-defines $I^B_2$, then there exists a complete injective $(C,C)$-implication in $\sA$.
\end{lemma}

\subsection{Critical relations}

We adapt the notion of a critical relation from~\cite{Wrona:2020b} to our situation and use results from the previous sections to prove that any structure which satisfies the assumptions of~\Cref{thm:implsimple} and which is implicationally hard on injective instances pp-defines a critical relation.

\begin{definition}\label{def:kcritical}
Let $k\geq 2$, and let $\sA$ be a relational structure. Let $C,D\subseteq I_k^A$ be disjoint and pp-definable from $\sA$, let $V$ be a set of $k+1$ variables, and let $\tuple u,\tuple v\in I^V_k$ be such that $\scope(\tuple u)\cup\scope(\tuple v)=V$ and such that $u_1,v_1\notin \scope(\tuple u)\cap \scope(\tuple v)$.
We say that a pp-formula $\phi$ over the signature of $\sA$ with variables from $V$ is \emph{critical in $\sA$ over $(C,D,\tuple u,\tuple v)$} if all of the following hold:
\begin{enumerate}
    \item $\phi$ is a complete $(C,\tuple u,C,\tuple v)$-implication in $\sA$,
    \item 
    $\phi^A$ contains every    $f \in A^{S(\tuple u) \cup S(\tuple v)}$ with $f(\tuple u) \in C$ and $f(\tuple v) \in C$
    and the same holds for $D$ in the place of $C$. 
\end{enumerate}
\end{definition}

\begin{lemma}\label{cor:critical}
    Let $k\geq 2$, let $\sB$ be $k$-neoliberal and pp-define $I^{B}_2$, and suppose that $\sB$ has finite duality. Suppose that $\relstr A$ is a CSP template which is a first-order expansion of $\sB$. Suppose moreover that $\sA$ is implicationally hard on injective instances. 
    Then there exists a pp-formula over the signature of $\sA$ which is critical in $\sA$ over $(C,D,\tuple u,\tuple v)$ for some $C,D\subseteq 
    A^k$.
\end{lemma}


\begin{proof}
Since the structure $\sA$ is implicationally hard on injective instances, the injective implication graph $\mathcal{G}_{\sA}^{\injinstances}$ contains a directed cycle $(D_1,C_1),\ldots,(D_{n-1},C_{n-1}),(D_n,C_n)=(D_1,C_1)$. This means that for all $i\in[n-1]$, there exists an injective $(C_i,\tuple u^i,C_{i+1},\tuple u^{i+1})$-implication $\phi_i$ in $\sA$ with $\proj_{\tuple u^i}(\phi_i^{\sA})=D_i$, and $\proj_{\tuple u^{i+1}}(\phi_i^{\sA})=D_{i+1}$.

Let us define $\phi:=((\phi_1\circ \phi_2)\circ\ldots\circ \phi_{n-1})$. Restricting $\phi^{\sA}$ to injective mappings, we obtain an injective $(C_1,\tuple u^1,C_1,\tuple u^n)$-implication by~\Cref{lemma:kimpl}. \Cref{lemma:kcomplete} asserts us that there exists a complete injective $(C_1,C_1)$-implication $\psi$ in $\sA$.

\Cref{observation:sinksource} yields that there exist $C\subseteq C_1$, and $D\subseteq D_1\backslash C_1$ such that $\mathbf{Vert}(C)$ is a strongly connected component which is a sink in $\mathcal{B}_{\psi}$, and $\mathbf{Vert}(D)$ is a strongly connected component which is a source in $\mathcal{B}_{\psi}$. 
Observe that since $\sA$ is a first-order expansion of $\sB$ and since $\psi$ is complete, $C$ is pp-definable from $\sA$. Indeed, for any fixed orbit $O\subseteq C$ under $\Aut(\sB)$, $C$ is equal to the set of all orbits $P\subseteq C_1$ such that $\psi^{\sA}$ contains an $O P$-mapping. We can observe in a similar way that $D$ is pp-definable from $\sA$ as well. 
Moreover, $\psi$ is easily seen to be a complete $(C,C)$-implication in $\sA$.

Since $\sB$ has finite duality, there exists a number $d\geq 3$ such that for every finite structure $\sX$ in the signature of $\sB$, it holds that if every substructure of $\sX$ of size at most $d-2$ maps homomorphically to $\sB$, then so does $\sX$. 
Set $\rho:=\psi^{\circ d}$. Let $V$ be the set of variables of $\rho$. It follows from~\Cref{lemma:kimpl} that $\rho$ is
a $(C,\tuple u,C,\tuple v)$-implication in $\sA$ for some $\tuple u,\tuple v$.
\medskip

We are now going to prove that the second bullet in~\Cref{def:kcritical} holds for $\rho$  using the finite duality of $\sB$ and the completeness of $\psi$.
To this end,
let $f\in A^V$ be as specified there. 
Up to renaming variables, we can assume that $\psi$ is a $(C,\tuple u,C,\tuple v)$-implication. 
Completeness of $\psi$ implies that $I_\psi=\{i\in[k]\mid u_i\in\scope(\tuple v)\}=\{i\in[k]\mid v_i\in\scope(\tuple u)\}$, and that $u_i=v_i$ for every $i\in I_\psi$. Let $\tuple w^0,\ldots,\tuple w^d$ be $k$-tuples of variables such that $w^j_i=u_i=v_i$ for all $i\in I_\psi$ and all $j\in\{0,\ldots,d\}$, and disjoint otherwise. We can moreover assume that $\tuple w^0=\tuple u$, $\tuple w^d=\tuple v$. 
For every $j\in\{0,\ldots,d-1\}$, let $\psi_{j+1}$ be the $(C,\tuple w^j,C,\tuple w^{j+1})$-implication obtained from $\psi$ by renaming $\tuple u$ to $\tuple w^j$ and $\tuple v$ to $\tuple w^{j+1}$. 
It follows that $\rho$ is equivalent to the formula obtained from $\psi_1\wedge\dots\wedge \psi_d$ by existentially quantifying all variables that are not contained in $\scope(\tuple u)\cup\scope(\tuple v)$. 
In order to prove that the required propert of a critical relation holds, it is therefore enough to show that $f$ can be extended to a mapping $h\in (\psi_1\wedge\dots\wedge\psi_d)^{\sA}$.

Now, for all $p,q\in[k]$, we identify the variables $w^0_{q}=u_q$ and $w^{d}_{p}=v_p$ if $f(v_{p}) = f(u_{q})$. Observe that $f|_{\scope(\tuple u)\cap \scope(\tuple v)}$ is injective since $f(\tuple u)\in C\subseteq I^A_k$, and hence this identification does not force any variables from $\scope(\tuple u)\cap\scope (\tuple v)$ to be equal. Moreover, since $d\geq 2$, this identification does not identify any variables from the tuples $\tuple w^{d-1}$ and $\tuple w^d$. Let us define $f_0:=f|_{\scope(\tuple w^0)\cup \scope(\tuple w^d)}$. 
It is enough to show that $f_0$ can be extended to a mapping $h\in (\psi_1\wedge \dots\wedge \psi_d)^{\sA}$.

Let $O$ be the orbit of $f(\tuple u)$ under $\Aut(\sB)$, and let $P$ be the orbit of $f(\tuple v)$. Let $f_1\in \psi_1^{\sA}$ be an injective $O P$-mapping, and for every $j\in\{2,\ldots,d\}$, let $f_j\in \psi_j^{\sA}$ be an injective $P P$-mapping. 
Note that such $f_j$ exists for every $j\in[d]$ since $\psi_j$ is complete and since $\mathbf{Vert}(C)$ and $\mathbf{Vert}(D)$ are strongly connected components in $\mathcal{B}_{\psi_j}$. 

Let $\tau$ be the signature of $\sB$. Let $X:=\scope(\tuple w^0)\cup\cdots\cup \scope(\tuple w^d)$, and let us define a $\tau$-structure $\sX$ on $X$ as follows. Recall that the relations from $\tau$ correspond to the orbits of injective $k$-tuples under $\Aut(\sB)$. For every relation $R\in\tau$, we define $R^{\sX}$ to be the set of all tuples $\tuple w\in I^X_k$ such that there exists $j\in \{0,\ldots,d\}$ such that $\scope(\tuple w)\subseteq \scope(\tuple w^j)\cup\scope(\tuple w^{j+1})$ and $f_j(\tuple w)\in R$; here and in the following, the addition $+$ on indices is understood modulo $d+1$. We will show that $\sX$ has a homomorphism $h$ to $\sB$. If this is the case, it follows by the construction and by the $k$-homogeneity of $\Aut(\sB)$ that $h\in (\psi_1\wedge\dots\wedge\psi_d)^{\sA}$. Moreover, we can assume that $h|_{\scope(\tuple w^0)\cup\scope(\tuple w^d)}=f_0$ as desired.

Let now $Y\subseteq X$ be of size at most $d-2$. 
Then for cardinality reasons there exists $j\in[d-1]$ such that $Y\subseteq \scope(\tuple w^0)\cup\cdots\cup \scope(\tuple w^{j-1})\cup \scope(\tuple w^{j+1})\cup\cdots\cup \scope(\tuple w^{d})$. 
Observe that $f_{j+1}$ is a homomorphism from the induced substructure of $\sX$ on $\scope(\tuple w^{j+1})\cup \scope(\tuple w^{j+2})$ to $\sA$. 
Since the orbits of $f_{j+1}(\tuple w^{j+2})$ and $f_{j+2}(\tuple w^{j+2})$ agree by definition, by composing $f_{j+2}$ with an element of $\Aut(\sB)$ we can assume that $f_{j+1}(\tuple w^{j+2})=f_{j+2}(\tuple w^{j+2})$. 
We can proceed inductively and extend $f_{j+1}$ 
to $\scope(\tuple w^0)\cup\cdots\cup \scope(\tuple w^{j-1})\cup \scope(\tuple w^{j+1})\cup\cdots\cup \scope(\tuple w^{d})$ such that it is a homomorphism from the induced substructure of $\sX$ on this set to $\sA$. It follows that the substructure of $\sX$ induced on $Y$ maps homomorphically to $\sA$. Finite duality of $\sB$ yields that $\sX$ has a homomorphism to $\sA$ as desired. It follows that $\rho$ has the property in the second bullet in\Cref{def:kcritical}.

Assume without loss of generality that $1\notin I_\rho$, and identify $u_i$ with $v_i$ for every $i\neq 1$. Note that this is possible by item (3) in \Cref{observation:impl_properties}, and since $\rho$ is easily seen to be complete.
Set $W:=\scope(\tuple u)\cup\scope(\tuple v)$, and let $\rho'$ be the formula arising from $\rho$ by this identification.
Observe that for every orbit $O\subseteq C$ under $\Aut(\sB)$, $(\rho')^{\sA}$ contains an injective $OO$-mapping. Indeed, there exists an injective $g\in A^W$ such that $g(\tuple u)\in O$ and $g(\tuple v)\in O$; this easily follows by the $(k-1)$-transitivity of $\Aut(\sB)$ and by the fact that it has no $k$-algebraicity. Hence 
$g\in\rho^{\sA}$, whence $g\in (\rho')^{\sA}$. 
It follows that $\rho'$ satisfies the second bullet in~\Cref{def:kcritical}. 
We will now argue that $\rho'$ is a $(C,\tuple u,C,\tuple v)$-implication in $\sA$. It immediately follows that $\rho'$ satisfies the items (1)-(3) and (5) from \Cref{def:implication}. Moreover, the satisfaction of item (4) follows immediately from the fact that $\rho$ is a $(C,\tuple u,C,\tuple v)$-implication in $\sA$.
It completes the proof that $\rho'$ is critical in $\sA$ over $(C,D,\tuple u,\tuple v)$ and hence the proof of the lemma.

\end{proof}

\subsection{Proof of \texorpdfstring{\Cref{thm:implsimple}}{Theorem~\ref{thm:implsimple}}}

We prove that no structure which satisfies the assumptions of~\Cref{thm:implsimple} can pp-define a critical relation. Using the results from the previous section, we get that no such structure can be implicationally hard on injective instances which immediately yields~\Cref{thm:implsimple}.
We start from $k$-neoliberal structures with $k \geq 3$ and reach to the case with $k=2$ later.

\begin{lemma}\label{lemma:kcritical}
Let $k\geq 3$, and let $\sA$ be a first-order expansion of a $k$-neoliberal relational structure $\sB$. Suppose that there exists a pp-formula $\phi$ which is critical in $\sA$ over $(C,D,\tuple u,\tuple v)$ for some $k$-ary $C,D$, and some $\tuple u,\tuple v$. Then $\sA$ is not preserved by any chain of quasi J\'{o}nsson operations.
\end{lemma}

\begin{proof}
First of all, observe that the formula $\phi_{\injinstances}:=\phi\wedge I_{k+1}$ is equivalent to a pp-formula over $\sA$ by the $2$-transitivity of $\Aut(\sB)$ and is still critical in $\sA$ over $(C,D,\tuple u,\tuple v)$.

Let $V=\{x_1,\ldots,x_{k+1}\}$ be the set of variables of $\phi$. We assume without loss of generality that 
$\tuple v=(x_2, \ldots,x_k, x_{k+1})$ and $\tuple u = (x_1, x_{3}, \ldots, x_{k+1})$.

Suppose that $\sA$ is preserved by a chain of quasi J\'{o}nsson operations of length $2n+1$, see~\Cref{def:QuasiJonsson}. 
Striving for contradiction, a number of times we will consider vectors
$\tuple w^2,\ldots,\tuple w^{k+1}\in A^3$ such that 
$\tuple w^3,\ldots,\tuple w^{k+1}$ are constant.
The goal is to show that there exist
$\tuple w^2, \tuple w^3, \ldots, \tuple w^{k+1}\in A^3$
such that, this time, all of them are constant and such that we have both 
 $(w^2_i,\ldots, w^{k+1}_i) \in D$ for $i \in [3]$
and $(J_{2n+1}(\tuple w^2),\ldots, J_{2n+1}(\tuple w^{k+1})) \in C$. It contradicts the fact that $D$ is pp-definable in $\sA$ and that $C \cap D = \emptyset$.

We will say that $\tuple w^2,\tuple w^3, \ldots,\tuple w^{k+1}\in A^3$ as specified above (so in particular all of them but the first one are constant) are of shape 
\begin{itemize}
\item $(c,c,d)$ if 
$(w^2_1,\ldots,w^{k+1}_1) = (w^2_2,\ldots,w^{k+1}_2) \in C$

and $(w^2_3,\ldots,w^{k+1}_3) \in D$, and of shape 
\item $(c,d,d)$  if $(w^2_1,\ldots,w^{k+1}_1) \in C$
and 
$(w^2_2,\ldots,w^{k+1}_2) = $

$(w^2_3,\ldots,w^{k+1}_3) \in D$
\end{itemize}
Since $\sB$ is $(k-1)$-transitive,  sequences of vectors of both shapes do exist. 

The proof of the following claim is a huge part of the proof of the lemma.

\begin{claim}\label{claim:hlp}
Let $\tuple w^2,\ldots,\tuple w^{k+1}\in A^3$ be of shape $(c,c,d)$ or of shape $(c,d,d)$. Then for all $i \in [2n+1]$ 
we have that  $(J_i(\tuple w^2),\ldots, J_i( \tuple w^{k+1})) \in C$.
\end{claim}

\begin{proof}
The proof goes by  induction on $i$. 

(BASE CASE) For $i = 0$ consider first $\tuple w^2,\ldots,\tuple w^{k+1}$ of shape 
$(c,c,d)$ and assume without loss of generality that $J_1$ is idempotent on the elements of the vectors, see~\Cref{obs:idemJonsson}. By~(\ref{eq:J1}), we have that $(J_1(\tuple w^2),\ldots, J_1(\tuple w^{k+1})) = (w^2_1,\ldots, w^{k+1}_1)$, and hence the claim follows.

Before we turn to the proof of the base case for the vectors of the shape $(c,d,d)$, we observe that the set $\phi_{\injinstances}^{\sA}$ can be viewed as a $(k+1)$-ary relation $R(x_1,\ldots,x_{k+1})$. Now, by $k$-neoliberality of $\sB$, there exists $\tuple w^1 \in A^3$
such that all of the following hold:
\begin{itemize}
\item  $\tuple w^1, \tuple w^3\ldots,\tuple w^{{k+1}}$ is of shape $(c,c,d)$.
\item $(w^1_i,\ldots,w^{k+1}_i)\in R$ for  $i\in \{ 1, 3\}$, and
\item 
        $w_2^1\neq w^{2}_2$.
\end{itemize}
Indeed, since $\sB$ $(k-1)$-transitive and has no $k$-algebraicity, there exists $w^1_2$ such that $(w^1_2, w^{3}_2, \ldots,w^{k+1}_2) \in C$ and $w^1_2$ does not appear in $\tuple w^2,\ldots,\tuple w^{k+1}$. Hence, in particular, the last bullet is satisfied. In order to satisfy the first and the second condition, we first set $w^1_1 := w^2_1$, and hence $(w^1_1, w^3_1, \ldots,w^{k+1}_1) \in C$. By~(2) in Definition~\ref{def:kcritical}, it follows that  
$(w^1_1, w^2_1, \ldots, w^{k+1}_1)\in R$. Furthermore, by $(k-1)$-transitivity and the second bullet in the definition of the critical relation, we also get $w^1_3$  such that $(w_3^1, w_3^3 \ldots,w_3^{k+1}) \in D$ and $(w^1_3,w^2_3, \ldots, w_{3}^{k+1})\in R$.
Clearly, the vectors are of shape $(c,c,d)$, whence they satisfy the first condition above. 
By the first part of the proof for the base case, we have $(J_1(\tuple w^1), J_1(\tuple w^3), \ldots, J_1(\tuple w^{k+1})) \in C$.

Now, since $k > 2$, by $(k-1)$-transitivity of $\sB$, there exist $a_3, \ldots, $
$a_{k+1}$ such that $(w^1_2, w^2_2, a_3, \ldots, a_{k+1}) \in R$.
Set $\tuple q^i = (w^i_1, a_i , w^i_3)$ for $i \in \{ 3, \ldots, k+1\}$.
By~(\ref{eq:Ji}), it follows that $(J_1(\tuple w^1), J_1(\tuple q^{3}), \ldots, J_1(\tuple q^{k+1}))$
$ = (J_1(\tuple w^1), J_1(\tuple w^{3}), \ldots, J_1(\tuple w^{{k+1}})) \in C$. We again assume that $J_1$ is idempotent on the considered elements.
Since $\phi_{\injinstances}$ is a $(C, \tuple u, C, \tuple v)$-implication, we obtain that $(J_1(\tuple w^2), J_1(\tuple q^3), \ldots, J_1(\tuple q^{k+1})) =  (J_1(\tuple w^2),$
$ \ldots, J_1(\tuple w^{k+1})) \in C.$ It completes the proof of the base case. 

(INDUCTION STEP) For the induction step, we consider two cases of whether $i$ is odd or even. If $i = 2k+1$ is odd, the claim holds for vectors of shape $(c,c,d)$ by the induction hypothesis and~(\ref{eq:J2i2i+1}). The proof for vectors of shape $(c,d,d)$
goes along the lines of the proof for $J_1$ in the base case.

On the other hand, if $i$ is even, the case of vectors of shape $(c,d,d)$ follows from the induction hypothesis by~(\ref{eq:J2i-12i}).
Let $\tuple w^2, \ldots, \tuple w^{k+1}$ be of shape $(c,c,d)$.
Again, by $k$-neoliberality of $\sB$, it follows that there
exists $\tuple w^1 \in A^3$
such that all of the following hold:
\begin{itemize}
\item  $\tuple w^1, \tuple w^{3}, \ldots,\tuple w^{k+1}$ is of shape $(c,d,d)$.
\item $(w^1_i,\ldots,w^{k+1}_i)\in R$ for  $i\in \{ 1, 3\}$, and
\item 
        $w^1_2\neq w^{2}_2$.
\end{itemize}
Indeed, since $\sB$ is  $(k-1)$-transitive and has no $k$-algebraicity, there exists $w^1_2$ such that $(w^1_2, w^{3}_2, \ldots,w^{{k+1}}_2) \in D$ and $w^2_1$ does not appear in $\tuple w^2,\ldots,\tuple w^{k+1}$. It follows that the last condition is satisfied. In order to satisfy the remaining bullets, we set $w^3_1 = w^2_1$. By 
the definition of a critical relation, $(w^3_1,\ldots,w^3_{k+1}) \in R$. It is also possible to find $w^1_1$, by $(k-1)$-transitivity, such that $(w^1_1, w^{3}_1, \ldots,w^{{k+1}}_1) \in C$ and $(w^3_1,\ldots,w^3_{k+1}) \in R$. Thus, all the conditions above are satisfied. Since we are done for the vectors of shape
$(c,c,d)$, it follows that $(J_{i}(\tuple w^1), \tuple w^{3}), \ldots, J_{i}(\tuple w^{{k+1}})) \in C$. 

Since $k > 2$, by $(k-1)$-transitivity of the structure $\sB$, there exist $a_3, \ldots, a_{k+1}$ such that $(w^1_2, w^2_2, a_3, \ldots, a_{k+1})$ is a tuple in $R$.
Set $\tuple q^i = (w^i_1, a_i , w^i_3)$ for $i \in \{ 3, \ldots, k+1\}$.
By~(\ref{eq:Ji}), it follows that $(J_i(\tuple w^1),  J_i(\tuple q^{3}),$
$ \ldots, J_i(\tuple q^{{k+1}})) = (J_i(\tuple w^1),  J_i(\tuple w^{3}), \ldots, J_i(\tuple w^{{k+1}})) \in C$. As in the base case, we assume that $J_i$ is idempotent on the considered vectors.  Since $\phi_{\injinstances}$ is a $(C, \tuple u, C, \tuple v)$-implication, we have that $(J_i(\tuple w^2), J_i(\tuple q^{3}), \ldots, J_i(\tuple q^{k+1})) =  (J_i(\tuple w^2), J_i(\tuple w^{3}) \ldots, J_i(\tuple w^{k+1})) \in C$. It completes the proof of the claim.  
\end{proof}

The final step of the proof of the lemma is to show that 
$(J_{2n+1}(\tuple w^2),$
$ \ldots, J_{2n+1}(\tuple w^{k+1})) \in C$
for the vectors satisfying  
$(w^2_i, \ldots, w^{k+1}_i) = (w^2_j, \ldots, w^{k+1}_j) \in D$
for all $i,j \in [3]$. Again by $k$-neoliberality, we can show that there exists
$\tuple w_1$ satisfying all of the following:
\begin{itemize}
\item  $\tuple w^1, \tuple w^{3}, \ldots, \tuple w^{{k+1}}$ is of shape $(c,d,d)$.
\item $(w^1_i,\ldots,w^{k+1}_i)\in R$ for  $i\in \{ 2, 3\}$, and
\item 
        $w^1_1\neq w^{2}_1$.
\end{itemize}

In the same way as above, we use $(k-1)$-transitivity of $\sB$ to find vectors $\tuple q^i = (a_i, w^2_i, w^3_i)$
with $i \in \{3, \ldots, k+1 \}$ such that $(w^1_1, w^2_1, q^3_1, \ldots, q^{k+1}_1) \in R$. By Claim~\ref{claim:hlp} and item~(\ref{eq:J2n+1}) in Definition~\ref{def:QuasiJonsson}, it follows that  $(J_{2n+1}(\tuple w^1), J_{2n+1}(\tuple q^{3}) \ldots, J_{2n+1}(\tuple q^{{k+1}}) ) = (J_{2n+1}(\tuple w^1), J_{2n+1}(\tuple w^{3}) \ldots, J_{2n+1}(\tuple w^{{k+1}}) )  \in C$. Since $\phi_{\injinstances}$ is a $(C, \tuple u, C,$
$ \tuple v)$-implication, it follows that 
$(J_{2n+1}(\tuple w^2), J_{2n+1}(\tuple q^{3}),  \ldots, J_{2n+1}(\tuple q^{k+1}$
$)) = (J_{2n+1}(\tuple w^2), J_{2n+1}(\tuple w^{3}),  \ldots, J_{2n+1}(\tuple w^{k+1})) \in C$, which contradicts the fact that $D$ is pp-definable in $\sA$ and disjoint from $C$.

\end{proof}

The case where  $\sB$ is $2$-neoliberal is slightly more complicated. We first show that such $\sB$ pp-defines disequality.

\begin{restatable}{lemma}{diseqJonsson}\label{lemma:diseqJonsson}
Let $\sB$ be a $2$-neoliberal structure preserved by a chain of quasi J\'{o}nsson operations. Then $\sB$ pp-defines $I_2^B$.
\end{restatable}

We  can now turn to a counterpart of~\Cref{lemma:kcritical} for $2$-neoliberal structures and then to the proof of~\Cref{thm:implsimple}.

\begin{restatable}{lemma}{2critical}\label{lemma:2critical}
Let $\sA$ be a first-order expansion of a $2$-neoliberal relational structure $\sB$. Suppose that there exists a pp-formula $\phi$ which is critical in $\sA$ over $(C,D,\tuple u,\tuple v)$ for some binary $C,D$, and some $\tuple u,\tuple v$. Then $\sA$ is not preserved by any chain of quasi J\'{o}nsson operations.
\end{restatable}

\implsimple*

\begin{proof}
Striving for contradiction, suppose that $\sA$ is implicationally hard on injective instances. Furthermore observe that $\sB$ pp-defines  $I^B_2$. Indeed. if it is $2$-neoliberal then the claim follows by~\Cref{lemma:diseqJonsson}. Otherwise   
 $\sB$ is $k$-neoliberal with $k \geq 3$, in particular $2$-transitive and hence $I^B_2$ is also pp-definable.
By~\Cref{cor:critical}, there exists a pp-formula $\phi$ over the signature of $\sA$ which is critical in $\sA$ over $(C,D,\tuple u,\tuple v)$ for some $C,D\subseteq A^k$. \Cref{lemma:kcritical,lemma:2critical} yield that $\sA$ is not preserved by any chain of quasi J\'{o}nsson operations, a contradiction.
\end{proof}

\kmain*

\begin{proof}
    Let $\sA$ be a first-order expansion of $\sB$ which is preserved by a chain of quasi J\'{o}nsson operations.
    By~\Cref{cor:libcores-purelyinj}, it is enough to prove that $\Csp_{\injinstances}(\sA)$ has relational width $(k,\max(k+1,b_{\sB}))$.
    \Cref{thm:implsimple} yields that $\sA$ is implicationally simple on injective instances and the result follows from~\Cref{prop:implsimple}.
\end{proof}

\bibliographystyle{IEEEtran}
\bibliography{mybib}

\appendix

\section{Omitted Proofs}

\subsection{Proof of \texorpdfstring{\Cref{obs:idemJonsson}}{Observation~\ref{obs:idemJonsson}}}

Since $\sA$ is a model complete core, for every finite $B \subseteq A$ and every  unary operation $J_i(x,x,x)$ with $i \in [2n+1]$ there exists an automorphism $\alpha_i$ such that $\alpha_i$ is equal to $J_i(x,x,x)$ on $B$. Now $J'_i(x,y,z) = \alpha_i^{-1}(J_i(x,y,z))$ with $i \in [2n+1]$ is a chain of quasi J\'{o}nsson operations idempotent on $B$.
$\qed$
\subsection{Proof of \texorpdfstring{\Cref{lemma:core-bwidth}}{Lemma~\ref{lemma:core-bwidth}}}

Let $\sim$ be a binary relation defined on $\V$ such that $u\sim v$ if, and only if, $\proj_{(u,v)}(\instance)\subseteq \{(b,b)\mid b\in B\}$. Since $k\geq 2$, $\sim$ is an equivalence relation.

Let $\tau$ be the signature of the structure $\sB$, and let us define a $\tau$-structure $\sA$ on $\V/\sim$ as follows. Let $R\in\tau$; then $R$ is of arity $k$. We set $([v_1]_{\sim},\ldots,[v_k]_{\sim})\in R^{\sA}$ if $\proj_{(v_1,\ldots,v_k)}(\instance)= R^{\sB}$. Note that by our assumption, for every $(v_1,\ldots,v_k)\in\V^k$, there is precisely one relation of $\sA$ containing the tuple $([v_1]_{\sim},\ldots,[v_k]_{\sim})$.

Let us show that the definition of $\sim$ does not depend on the choice of the representatives $v_1,\ldots,v_k\in\V$. We will show that it does not depend on the choice of $v_1$, the rest can be shown similarly. Let $u_1\sim v_1$, and let $C\in\constraints$ be such that $u_1,v_1,\ldots,v_k$ are contained in its scope. Then $C|_{\{u_1,v_1\}}$ consists of constant maps and it follows that $\proj_{(u_1,v_2,\ldots,v_k)}(\instance)=\proj_{(u_1,v_2,\ldots,v_k)}(C)=\proj_{(v_1,\ldots,v_k)}(C)=\proj_{(v_1,\ldots,v_k)}(\instance)$.

We claim that $\sA$ embeds into $\sB$. Suppose for contradiction that this is not the case. Then there exists a bound $\sF\in\mathcal{F}_{\sB}$ of size $b$ with $b\leq b_{\sB}$ such that $\sF$ embeds into $\sA$. Let $[v_1]_{\sim},\ldots,[v_b]_{\sim}$ be all elements in the image of this embedding. Find a constraint $C\in\constraints$ such that $v_1,\ldots,v_b$ are contained in its scope. Since $C$ is nonempty, there exists $f\in C$. Since all relations in $\tau$ are of arity $k$, since $\instance$ is $k$-minimal and since for every $\tuple v\in \V^k$ such that all variables from $\tuple v$ are contained in the scope of $C$, $\proj_{\tuple v}(C)$ contains precisely one orbit under $\gG$, it follows that $\sF$ embeds into the structure that is induced by the image of $f$ in $\sB$ which is a contradiction.

It follows that there exists en embedding $e\colon \sA\hookrightarrow \sB$ and it is easy to see that $f\colon \V\rightarrow B$ defined by $f(v):=e([v]_{\sim})$ is a solution of $\instance$.
$\qed$

\subsection{Proof of \texorpdfstring{\Cref{prop:implsimple}}{Proposition~\ref{prop:implsimple}}}

By~\Cref{cor:libcores-purelyinj}, it is enough to show that $\Csp_{\injinstances}(\sA)$ has relational width $(k,\max(k+1,b_{\sB}))$. To this end, let $\instance=(\V,\constraints)$ be a non-trivial $(k,\max(k+1,b_{\sB}))$-minimal instance of $\Csp_{\injinstances}(\sA)$; we will show that there exists a satisfying assignment for $\instance$.

For every $i\geq 0$, we define inductively a $(k,\max(k+1,b_{\sB}))$-minimal instance $\instance_i=(\V,\constraints_i)$ of $\Csp_{\injinstances}(\sA)$ with the same variable set as $\instance$ such that $\instance_0=\instance$ and such that for every $i\geq 1$, $\constraints_i$ contains for every constraint $C_{i-1}\in\constraints_{i-1}$ a constraint $C_i$ such that $C_i\subseteq C_{i-1}$, and such that moreover for some $\tuple v\in\V^k$, it holds that $\proj_{\tuple v}(\instance_i)\subsetneq \proj_{\tuple v}(\instance_{i-1})$, or for every $\tuple v\in\V^k$, it holds that $\proj_{\tuple v}(\instance_i)$ contains only one orbit under $\fG$.

Let $\instance_0:=\instance$. Let $i\geq 1$. We define $\mathcal{G}_i$ to be the graph that originates from $\mathcal{G}_{\sA}^{\injinstances}$ by removing all vertices that are not of the form $(\proj_{\tuple v}(\instance_{i-1}),F)$ for some injective $\tuple v\in \V^k$, and some $F\subseteq A^k$. \Cref{claim:projs} implies that $\instance_{i-1}$ is $k$-minimal, and hence $\proj_{\tuple v}(\instance_{i-1})$ is well-defined. If $\mathcal G_i$ does not contain any vertices, let $\instance_i:=\instance_{i-1}$. Suppose now that $\mathcal G_i$ contains at least one vertex. In this case, since $\mathcal{G}_{\sA}^{\injinstances}$ and hence also $\mathcal G_i$ is acyclic, we can find a sink $(\proj_{\tuple v_i}(\instance_{i-1}),F_i)$ in $\mathcal{G}_i$ for some injective $\tuple v_i\in\V^k$. Let us define for every $C_{i-1}\in\constraints_{i-1}$ containing all variables from $\tuple v_i$ in its scope $C_{i}:=\{f\in C_{i-1}\mid f(\tuple v_i)\in F_i\}$, and let $C_i:=C_{i-1}$ for every $C_{i-1}\in\constraints_{i-1}$ that does not contain all variables from $\tuple v_i$ in its scope. Note that in both cases, $C_i\subseteq C_{i-1}$.
Finally, we define $\constraints_{i}=\{C_{i}\mid C_{i-1}\in\constraints_{i-1}\}$.

\begin{claim}\label{claim:projs}
For every $i\geq 1$, $\instance_i$ is non-trivial and $(k,\max(k+1,b_{\sB}))$-minimal. Moreover, for every $\tuple v\in \V^k\backslash\{\tuple v_i\}$, $\proj_{\tuple v}(\instance_{i})=\proj_{\tuple v}(\instance_{i-1})$ and $\proj_{\tuple v_i}(\instance_{i})=F_i$.
\end{claim}

Let $i\geq 1$ and if $i>1$, suppose that the claim holds for $i-1$. Note that if $\instance_i=\instance_{i-1}$, then there is nothing to prove so we may suppose that this is not the case. Observe that for every $C_i\in\constraints_i$ containing all variables from $\tuple v_i$ in its scope, $\proj_{\tuple v_i}(C_i)=F_i$ by the definition of $C_i$.
We will now show that for every $\tuple v\in \V^k\backslash\{\tuple v_i\}$ and for every $C_i\in\constraints_i$ containing all variables from $\tuple v$ in its scope, $\proj_{\tuple v}(C_i)=\proj_{\tuple v}(C_{i-1})$. Observe that if $C_i$ does not contain all variables from $\tuple v_i$ in its scope, then the conclusion follows immediately since $C_i=C_{i-1}$; we can therefore assume that this is not the case.
Assume first that $\tuple v$ is not injective, let $m$ be the number of pairwise distinct entries of $\tuple v$, and let $\tuple u$ be an injective $m$-tuple containing all variables of $\tuple v$. Hence, $\proj_{\tuple u}(C_i)=I_m^A=\proj_{\tuple u}(C_{i-1})$ by the $(k-1)$-transitivity of $\gG$ and it follows that $\proj_{\tuple v}(C_i)=\proj_{\tuple v}(C_{i-1})$. Now assume that $\tuple v$ is injective  and, striving for a contradiction, that $\proj_{\tuple v}(C_i)\subsetneq \proj_{\tuple v}(C_{i-1})$. It follows that $(\proj_{\tuple v}(C_{i-1}),\proj_{\tuple v}(C_i))$ is a vertex in $\mathcal G_{\sA}^{\injinstances}$ and hence also in $\mathcal G_i$.
Let $\tuple w=(w_1,\ldots,w_\ell)\in\V^\ell$ be an enumeration of all variables of $\tuple v$ and $\tuple v_i$. It follows that the pp-formula $\phi(\tuple w)$ defining $\proj_{\tuple w}(C_{i-1})$ is an injective $(F_i,\tuple v_i,\proj_{\tuple v}(C_i), \tuple v )$-implication. Hence, there is an arc from $(\proj_{\tuple v_i}(\instance_{i-1}),F_i)$ to $(\proj_{\tuple v}(C_{i-1}),\proj_{\tuple v}(C_i))$ in $\mathcal{G}_i$ and in particular, $(\proj_{\tuple v_i}(\instance_{i-1}),F_i)$ is not a sink in $\mathcal G_i$, a contradiction.

Now, it is easy to see that $\instance_i$ is $(k,\max(k+1,b_{\sB}))$-minimal. Indeed, since $\instance$ is $(k,\max(k+1,b_{\sB}))$-minimal, every subset of $\V$ of size at most $\max(k+1,b_{\sB})$ is contained in the scope of some constraint of $\instance$ and by construction also of $\instance_i$. Moreover, by the previous paragraph, any two constraint of $\instance_i$ agree on all $k$-element subsets of $\V$ within their scopes.

Since for every $i\geq 0$, if $\mathcal{G}_i$ is not empty, we remove at least one orbit of $k$-tuples under $\fG$ from some constraint. By the oligomorphicity of $\gG$, there exists $i_0\geq 0$ such that $\mathcal G_{i_0}$ is empty. 
We claim that for every injective $\tuple v\in \V^k$, $\proj_{\tuple v}(\instance_{i_0})$ contains precisely one orbit of $k$-tuples under $\fG$: If $\proj_{\tuple v}(\instance_{i_0})$ contained more than one orbit, then $(\proj_{\tuple v}(\instance_{i_0}),O)$ would be a vertex of $\mathcal{G}_i$ for an arbitrary orbit $O\subseteq\proj_{\tuple v}(\instance_{i_0})$; $O$ being a relation of $\sA$ since $\sA$ is a first-order expansion of $\sB$.

Thus $\instance_{i_0}=(\V,\constraints_{i_0})$ is a non-trivial, $(k,\max(k+1,b_{\sB}))$-minimal instance of $\Csp_{\injinstances}(\sB)$ that satisfies the assumptions of~\Cref{lemma:core-bwidth}. Hence, there exists a satisfying assignment for $\instance_{i_0}$ and whence also for $\instance$.
$\qed$

\subsection{Proof of \texorpdfstring{\Cref{prop:libcores-purelyinj}}{Proposition~\ref{prop:libcores-purelyinj}}}

    If $\sA$ has relational width $(k,\ell)$, then so does $\Csp_{\injinstances}(\sA)$ since every instance of $\Csp_{\injinstances}(\sA)$ is an instance of $\Csp(\sA)$.

    Suppose now that $\Csp_{\injinstances}(\sA)$ has relational width $(k,\ell)$.
    Let $f$ be an injective binary polymorphism of $\sA$. It follows that for every $\tuple s\in A^2$, $\tuple t\in I^2_A$, it holds that $f(\tuple s,\tuple t), f(\tuple t,\tuple s)\in I^2_A$. Let $\instance=(\V,\constraints)$ be a $(k,\ell)$-minimal non-trivial instance equivalent to an instance of $\Csp(\sA)$; we will show that $\instance$ has a solution. To this end, let $\instance'$ be an instance obtained from $\instance$ by identifying all variables $u,v\in \V$ with $\proj_{(u,v)}(\instance)\subseteq \{(a,a)\mid a\in A\}$, and by intersecting every constraint containing $u,v$ with  $C_{(u,v)}:=\{f\in A^{\{u,v\}}\mid (f(u), f(v))\in I_2^A
    \}$ for every $u,v\in \V$ with $\proj_{(u,v)}(\instance)\not\subseteq \{(a,a)\mid a\in A\}$. It follows that $\instance'$ is a non-trivial instance of $\Csp_{\injinstances}(\sA)$ and moreover, every solution of $\instance'$ translates into a solution of $\instance$. Lemma 18 in~\cite{hypergraphs}
    yields that the $(k,\ell)$-minimal instance $\cJ$ equivalent to $\instance'$ is non-trivial and since $\instance'$ is an instance of $\Csp_{\injinstances}(\sA)$ which has relational width $(k,\ell)$, $\cJ$ has a solution which translates into a solution of $\instance$ as desired.
$\qed$

\subsection{Proof of \texorpdfstring{\Cref{cor:libcores-purelyinj}}{Corollary~\ref{cor:libcores-purelyinj}}}

Since $\sA$ is preserved by a chain of quasi J\'{o}nsson operations, 
it does not have a uniformly continuous clone homomorphism to the clone of projections. It is easy to see and well-known that $\Pol(\sA)$ then contains an essential function and by~\Cref{lemma:binessen} a binary essential operation. ~
Since $I_2^B$ is pp-definable from $\sB$ and the canonical binary structure of $\sB$ has finite duality, 
\Cref{lemma:bininj2} yields that $\Pol(\sA)$ contains a binary injection. Now, the statement follows directly from~\Cref{prop:libcores-purelyinj}.
$\qed$

\subsection{Proof of \texorpdfstring{\Cref{lemma:l+2=rel}}{Lemma~\ref{lemma:l+2=rel}}}
Let us enumerate the variables of $\phi$ by $x_1,\ldots,x_{\ell+2}$. Without loss of generality, $\tuple u=(x_1,\ldots,x_{\ell})$ and $\tuple v=(x_{\ell+1},x_{\ell+2})$ are such that $\phi$ is an $(T,\tuple u,=,\tuple v)$-implication in $\sA$. Note that $S(\tuple u)\cap S(\tuple v)=\emptyset$ by the definition of an implication and since $\phi$ has $\ell+2$ variables. The set $\phi^{\sA}$ can then be viewed as an $(\ell+2)$-ary relation $R(x_1,\ldots,x_{\ell+2})$.

Using the $k$-neoliberality and finite duality of $\sB$, we can find $n\geq 1$ and $S := \{a^i_j\in B\mid 1\leq i\leq n-1, 1\leq j\leq \ell + 1 \}
\cup \{b^i_j\in B\mid 1\leq i\leq n-1, 1\leq j\leq \ell+1 \} \cup \{ 
a^n_{\ell + 1}, b^n_{\ell + 1} \} \cup \{ c_1, d_1, \ldots, d_{\ell} \}
$ with $a_{\ell+1}^1 = b_{\ell+1}^n,a_{\ell+1}^n = b_{\ell+1}^1$  
such that all of the following hold:
\begin{itemize}
\item $(d_1,\ldots ,d_\ell) \in T$
\item $(d_1,\ldots ,d_\ell, a^1_{\ell+1}) \in I_{\ell+1}^B$
\item $(d_1,\ldots ,d_\ell, a^1_{\ell+1},a^1_{\ell+1})\in R$
\item $(c_1, d_2, \ldots ,d_\ell) \in T$
\item $(c_1, d_2, \ldots ,d_\ell, a^n_{l+1}, a^n_{l+1}) \in T$
\item $(a^i_1,\ldots ,a^i_\ell,a^i_{\ell+1},a^{i+1}_{\ell+1})\in R$ for $i \in [n-1]$,
\item $(b^i_1,\ldots ,b^i_\ell, b^i_{\ell+1},b^{i+1}_{\ell+1})\in R$ for $i \in [n-1]$.
\end{itemize}

To find these elements, let first $(d_1,\ldots ,d_\ell)\in T$ be arbitrary. We claim that there exists $a^1_{\ell+1}$ different from all $d_1,\ldots ,d_\ell$ such that $(d_1,\ldots ,d_\ell, a^1_{\ell+1},a^1_{\ell+1})\in R$. Indeed, observe that otherwise the projection of $\phi$ to its first $\ell+1$ coordinates is an $(O_D, =)$-implication where $O_D$ is the orbit of $d_1,\ldots ,d_\ell$. Hence, we are done by~\Cref{lemma:l+1=rel}. 
By the fact that $\sB$ has no $k$-algebraicity, we can find $a^n_{\ell+1}$ such that $a^n_{\ell+1} \neq a^1_{\ell+1}$, the choice of the value of $n$ is made later, and both these elements lie in the same orbit under the stabilizer of $\Aut(\sB)$ by $d_2,\ldots,d_{\ell}$. In particular, it follows that $(d_2,\ldots,d_\ell, a^n_{\ell+1})$ and $(d_2,\ldots,d_\ell,a^1_{\ell+1})$ lie in the same orbit under $\Aut(\sB)$, and hence there exists $c_1\in B$ such that $(c_1,d_2,\ldots,d_\ell,a^n_{\ell+1})$ and
$(d_1,d_2,\ldots,d_\ell,a^1_{\ell+1})$ lie in the same orbit under $\Aut(\sB)$. In particular, $(c_1,d_2,\ldots,d_\ell)\in T$ and it follows that $(c_1,d_2,\ldots,d_\ell,a^n_{\ell+1},a^n_{\ell+1})\in R$. 

 If $k \geq 3$ we set $n := 2$. In this case, by $2$-transitivity of $\sB$  we have $(a^1_{\ell+1},a^n_{\ell+1})\in I^B_2 \subseteq \proj_{(\ell+1,\ell+2)}(R)$, and hence there exist $a^1_1,\ldots,a^1_\ell\in B$ such that $(a^1_1,\ldots ,a^1_\ell,a^1_{\ell+1},a^2_{\ell+1})\in R$. Since
 $a^1_{\ell+1} = b^n_{\ell+1}, a^n_{\ell+1} =b^1_{\ell+1}$, by the same argument we have that there exist
 $b^1_1,\ldots,b^1_\ell \in B$ such that $(b^1_1,\ldots , b^1_\ell,b^1_{\ell+1},b^2_{\ell+1}) \in R$.


If $k= 2$, $\ell=2$, let $n$ be such that for every finite structure $\sX$ in the signature of $\sB$, it holds that if every substructure of $\sX$ of size at most $n-1$ maps homomorphically to $\sB$, then so does $\sX$. This is possible by the finite duality of $\sB$.
Then, in the same way as in the proof of~\Cref{lemma:l+1=rel},
for all $i \in [n-1]$ we find
$(a^i_1,a^i_2,a^i_3, a^{i+1}_3), (b^i_1,b^i_2,b^{i}_3,b^{i+1}_3) \in R$ as desired. 

We now  assume on the contrary that there exists a chain of J\'{o}nsson operations $(J_1, \ldots, J_{2m+1})$ satisfying~\Cref{def:QuasiJonsson}. By~\Cref{obs:idemJonsson} we may assume that the operations are idempotent on the set $S$.
We will now prove a claim similar to~\Cref{claim:equall+1}
and then make a final step which shows that $J_{2m+1}(a^n_{l+1}, a^n_{l+1}, a^n_{l+1}) = a^1_{l+1}$. It contradicts the assumption on the idempotency of $(J_1, \ldots, J_{2m+1})$ on $S$ and completes the proof of the lemma.

\begin{claim}
\label{claim:equall+2}
For all $k \in [2m+1]$ we have both of the following:
\begin{itemize}
\item $J_{k}(a^1_{\ell+1}, a^1_{\ell+1}, a^n_{\ell+1}) = a^1_{\ell+1}$
\item $J_{k}(a^1_{\ell+1}, a^n_{\ell+1}, a^n_{\ell+1}) = a^1_{\ell+1}$
\end{itemize}
\end{claim}

\begin{proof}
The proof goes by induction on $k \in [2m+1]$.

\textbf{(Base case: $k = 1$)}
For the first item we simply use the idempotency of $J_1$ on $S$ and~(\ref{eq:J1}) in~\Cref{def:QuasiJonsson}. In order to obtain the second item, by  induction on $i \in [n]$ we prove $J_1(a^1_{\ell+1}, a^i_{\ell+1}, a^n_{\ell+1}) = a^1_{\ell+1}$.
By the first item, we have that $J_1(a^1_\ell, a^1_\ell, a^n_\ell) = a^1_\ell$.  For the induction step, we assume that
 $J_1(a^1_\ell,a^{i}_{\ell+1}, a^1_{\ell+1}) = a^1_{\ell+1}$
for some $i\in[n- 1]$ and show
  $J_1(a^1_{\ell+1},a^{i+1}_{\ell+1}, a^1_{\ell+1}) = a^1_{\ell+1}$.

Put
\begin{align*}
   \tuple t_{1,1} := \begin{pmatrix}
       d_1 \\
       d_2 \\
       \ldots \\
       d_\ell \\
       a^1_{\ell+1} \\
       a^1_{\ell+1}
    \end{pmatrix},
    \tuple t_{i,i+1} := \begin{pmatrix}
       a^i_1 \\
       a^i_2 \\
       \ldots \\
       a^i_{\ell} \\
       a^i_{\ell+1} \\
       a^{i+1}_{\ell+1}
    \end{pmatrix},
      \tuple t_{n,n} := \begin{pmatrix}
       c_1 \\
       d_2 \\
       \ldots \\
       d_\ell \\
       a^n_{\ell+1} \\
       a^n_{\ell+1}
    \end{pmatrix}.
\end{align*}
for all $i \in [n-1]$.
By the discussion above, all these tuples are in $R$. 
Since $J_1$ preserves $R$, it follows that $$J_1(\tuple t_{1,1}, \tuple t_{i,i+1},\tuple t_{n,n})\in R$$ for all $i \in [n-1]$. This means that
\begin{align*}
    \begin{pmatrix}
         && J_1(c_1 ,a^i_1,d_1) \\
       d_2 &=& J_1(d_2 ,a^i_2,d_2) \\
       \ldots && \ldots \\
       d_{\ell} &=& J_1(d_{\ell},a^i_{\ell}, d_{\ell}) \\
       a^1_{\ell+1} &=& J_1(a^1_{\ell+1},a^i_{\ell+1},a^n_{\ell+1}) \\
    &&J_1(a^1_{\ell+1},a^{i+1}_{\ell+1},a^n_{\ell+1})
    \end{pmatrix} \in R.
\end{align*}
where all but the last of the displayed equalities follow by~(\ref{eq:Ji}) in Definition~\ref{def:QuasiJonsson} and the last one by the induction hypothesis. 
In order to demonstrate that $(J_1(a^1_{\ell+1},a^{i+1}_{\ell+1},a^n_{\ell+1}) = a^1_{\ell+1}$, we show that 
$(J_1(c_1 ,a^i_1,d_1), d_2, \ldots, d_{\ell}) \in T$. To this end consider  
\begin{align*}
    \tuple s_{cd}:=\begin{pmatrix}
       c_1 \\
       d_2 \\
       \ldots \\
       d_{\ell}
    \end{pmatrix},
    \tuple s_{ae}:=\begin{pmatrix}
       a^i_1 \\
       e_2 \\
       \ldots \\
       e_\ell
    \end{pmatrix},
    \tuple s_{d} :=\begin{pmatrix}
       d_1 \\
       d_2 \\
       \ldots \\
       d_\ell
    \end{pmatrix}.
\end{align*}
Indeed, since $\sB$ is transitive we can find $e_2, \ldots, e_{\ell}$ so that 
$\tuple s_{ae}$ is in $T$. The other two tuples are in $T$ by the discussion above. By~(\ref{eq:Ji}) in Definition~\ref{def:QuasiJonsson} we have 
that $J_1(\tuple s_{cd}, \tuple s_{ae}, \tuple s_{d}) = (J_1(c_1, 
a^i_1, d_1), d_2, \ldots, d_{\ell-1})$. The relation $T$ is pp-definable in $\sA$, and hence we obtain that the outcome tuple is also in $T$. In consequence, $J_1(a^1_{\ell+1},a^{i+1}_{\ell+1},a^n_{\ell+1}) = a^1_{\ell+1}$, and we are done with the base case for $k=1$.

\noindent
\textbf{(Induction step.)}
We now assume that the both items in the formulation of the claim hold for some $k \in [2m]$ and prove that both items hold for $k+1$. The proof depends on whether $k$ is even or odd. In the former case we have $J_{k+1}(a^1_{\ell+1}, a^1_{\ell+1}, a^n_{\ell+1}) = J_{k}(a^1_{\ell+1}, a^1_{\ell+1}, a^n_{\ell+1}) = a^1_1$ by~(\ref{eq:J2i2i+1}). 
To obtain 
$J_{k+1}(a^1_{\ell+1}, a^n_{\ell+1}, a^n_\ell) = a^1_{\ell+1}$ we proceed in the same way as in the base case for $J_1$. 

On the other hand, if $k$ is odd, then by~(\ref{eq:J2i-12i}) in~\Cref{def:QuasiJonsson}, we have $J_{k+1}(a^1_{\ell + 1}, a^n_{\ell+1}, a^n_{\ell+1}) = J_{k}(a^1_{\ell + 1}, a^n_{\ell+ 1}, a^n_{\ell+1}) = a^1_{l+1}$. By the assumption, 
it is the same as $J_{k+1}(a^1_{\ell+1}, b^1_{\ell+1}, a^n_{\ell+1}) = a^1_{{\ell+1}}$.
By the induction on $i \in [n]$ we will now show that $J_{k+1}(a^1_{\ell+1}, b^i_{\ell+1}, a^n_{\ell+1}) = a^1_{\ell+1}$. Since $b^n_{\ell+1}$ is $a^1_{\ell+1}$
we will have $J_{k+1}(a^1_{\ell+1}, a^1_{\ell+1}, a^n_{\ell+1}) = a^1_{\ell+1}$ as desired.
Assume that $J_{k+1}(a^1_{\ell+1}, b^i_{\ell+1}, a^n_{\ell+1}) = a^1_{\ell+1}$ for some $i$, we will show that $J_{k+1}(a^1_{\ell+1}, b^{i+1}_{\ell+1}, a^n_{\ell+1}) = a^1_{\ell+1}$.
To this end we observe that for all $i \in [n-1]$ the tuples

\begin{align*}
    \tuple t_{1,1} := \begin{pmatrix}
       d_1 \\
       d_2 \\
       \ldots \\
       d_\ell \\
       a^1_{\ell+1} \\
       a^1_{\ell+1}
    \end{pmatrix},
    \tuple t'_{i,i+1} := \begin{pmatrix}
       b^i_1 \\
       b^i_2 \\
       \ldots \\
       b^i_{\ell} \\
       b^i_{\ell+1} \\
       b^{i+1}_{\ell+1}
    \end{pmatrix},
   \tuple t_{n,n} := \begin{pmatrix}
       c_1 \\
       d_2 \\
       \ldots \\
       d_\ell \\
       a^n_{\ell+1} \\
       a^n_{\ell+1}
    \end{pmatrix},
\end{align*}
 are in $R$ by construction. 
Since $J_{k+1}$ preserves $R$, it follows that $$J_{k+1}(\tuple t_{1,1}, \tuple t'_{i,i+1},\tuple t_{n,n})\in R$$ for all $i \in [n-1]$. This means that
\begin{align*}
    \begin{pmatrix}
         && J_{k+1}(d_1  ,b^i_1, c_1) \\
       d_2 &=& J_{k+1}(d_2 ,b^i_2,d_2) \\
       \ldots && \ldots \\
       d_{\ell} &=& J_{k+1}(d_{\ell-1},b^i_{\ell}, d_{\ell-1}) \\
       a^1_{\ell+1} &=& J_{k+1}(a^1_{\ell+1},b^i_{\ell+1},a^n_{\ell+1}) \\
    &&J_{k+1}(a^1_{\ell+1},b^{i+1}_{\ell+1},a^n_{\ell+1})
    \end{pmatrix} \in R.
\end{align*}
where all but the last of the displayed equalities follow by~(\ref{eq:Ji}) in Definition~\ref{def:QuasiJonsson} and the last one by the induction hypothesis. 
In order to obtain $(J_{k+1}(a^1_{\ell+1},b^{i+1}_{\ell+1},a^n_{\ell+1}) = a^1_{\ell+1}$, it is enough to show that the tuple 
$(J_{k+1}(d_1 ,b^i_1,c_1), d_2, \ldots, d_{\ell})$ is in $T$. To this end consider  
\begin{align*}
    \tuple s_{d} :=\begin{pmatrix}
       d_1 \\
       d_2 \\
       \ldots \\
       d_\ell
    \end{pmatrix},
    \tuple s_{be}:=\begin{pmatrix}
       b^i_1 \\
       e_2 \\
       \ldots \\
       e_\ell
    \end{pmatrix},
     \tuple s_{cd}:=\begin{pmatrix}
       c_1 \\
       d_2 \\
       \ldots \\
       d_{\ell}
    \end{pmatrix}.
\end{align*}
Since $\sB$ is transitive we can find $e_2, \ldots, e_{\ell}$ so that 
$\tuple s_{be}$ is in $T$. The other two tuples are in $T$ by the discussion above. By~(\ref{eq:Ji}) in Definition~\ref{def:QuasiJonsson} we have 
that $J_{k+1}(\tuple s_{cd}, \tuple s_{be}, \tuple s_{d}) = (J_{k+1}(d_1, 
b^i_1, c_1), d_2, \ldots, d_{\ell})$. Since $T$ is pp-definable in $\sA$, the latter tuple is in $T$. In consequence we obtain $J_{k+1}(a^1_{\ell+1},b^{i+1}_{\ell+1},a^n_{\ell+1}) = a^1_{\ell+1}$. It completes the proof of the claim.
\end{proof}

By the claim we have $J_{2m+1}(a^1_{\ell+1}, a^n_{\ell+1}, a^n_{\ell+1}) = a^1_{\ell+1}$. In order to complete the proof of the lemma, by induction
on $i \in [n]$ we will show that $J_{2m+1}(a^i_{\ell+1}, a^n_{\ell+1}, a^n_{\ell+1}) = a^1_{\ell+1}$ for all $i \in [n]$. The base case is already done. 
For the induction step assume that  we have $J_{2m+1}(a^i_{\ell+1}, a^n_{\ell+1}, a^n_{\ell+1}) = a^1_{\ell+1}$ for some $i \in [n]$. 
Since $J_{2m+1}$ preserves $R$, it follows that $$J_{2m+1}(\tuple t_{i,i+1}, \tuple t_{n,n},\tuple t_{n,n})\in R,$$ for all $i \in [1, n-1]$. This means that
\begin{align*}
    \begin{pmatrix}
       c_1 &=& J_{2m+1}(a^i_1,c_1,c_1) \\
       d_2 &=& J_{2m+1}(a^i_2, d_2,d_2) \\
       \ldots && \ldots \\
       d_{\ell} &=& J_{2m+1}(a^i_{\ell},d_{\ell},d_{\ell}) \\
       a^1_{\ell+1} &=& J_{2m+1}(a^i_{\ell+1},a^n_{\ell+1},a^n_{\ell+1}) \\
    &&J_{2m+1}(a^{i+1}_{\ell+1},a^{n}_{\ell+1},a^n_{\ell+1})
    \end{pmatrix} \in R.
\end{align*}
where the first $l$ equalities follow by the idempotency of $J_{2m+1}$ on $S$ and~(\ref{eq:J2n+1}) in Definition~\ref{def:QuasiJonsson} and the last one by the induction hypothesis. 
Since the tuple $(d_1,\ldots,d_\ell)$ is in $T$ and $\phi$ is an $(T,=)$-implication, we obtain $J_{2m+1}(a^i_\ell,a^{n}_{\ell+1}, a^n_{\ell+1}) =  J_{2m+1}(a^{i+1}_{\ell+1},a^{n}_{\ell+1}, a^n_{\ell+1}) = a^{1}_{\ell+1}$, as desired. 
It completes the proof of the lemma.$\qed$

\subsection{Proof of~\Cref{cor:noequality}}

    Assume on the contrary that the required $h$ does not exist and that $\tuple v'$ is a maximal extension of $\tuple v$ for which there exists 
    $h\in\phi^{\sA}$ such that $h(\tuple v')$ is injective. Since $\sB$ is 
    $k$-homogeneous and contains all orbits of $k$-tuples there exists a pp-formula $\psi'$ over variables $\tuple v'$ which defines the orbit of $h(\tuple v')$. Observe that  $\tuple v'$ is of length at least $2$. If $k=2$, then $\psi'$ contains a single atom for every pair of variables. 
    If $k > 2$ and the length of $\tuple v'$ is equal to or greater than $k$, then $\psi'$ contains a single atom for every $k$-tuple of variables. Otherwise we simply write $v_1 \neq v_2$ for every pair of variables in $\tuple v'$. In any case, $\psi'$ is just a conjunction of some atoms built over symbols in the signature of $\sB$ and variables in $\tuple v'$. 
    Let now $\tuple v''$ be an extension to any variable $v_1$ in $V \setminus \scope(\tuple v')$. Clearly, there exists $v_2 \in \scope(\tuple v')$ such that the projection of $\phi \wedge \psi'$ to $(v_1, v_2)$ is $=$.
    Let now $\psi''$ be a maximal subformula of $\psi'$ such that this property holds, i.e. the projection of $\phi \wedge \psi''$ to $(v_1, v_2)$ is $=$. \chmichal{Then there exists an atom $O_{\tuple w}(\tuple w)$ with 
    $\tuple w = (w_1, \ldots, w_\ell)$ in $\psi''$
    and a subformula $\psi$ of $\psi''$
    such that $\psi'' := \psi \wedge O_{\tuple w}(\tuple w)$ 
    and that the projection of 
    $\phi \wedge \psi$ to $(w_1, \ldots, w_\ell, v_1, v_2)$ is
    a $(O_{\tuple w}, \tuple w, =, (v_1, v_2))$-implication.}
\Cref{lemma:l+1=rel,lemma:l+2=rel} now contradict the assumption that $\sA$
    is preserved by a chain of quasi J\'{o}nsson operations, which completes the proof of the corollary.$\qed$ 

\subsection{Proof of \texorpdfstring{\Cref{lemma:kimpl}}{Lemma~ref{lemma:kimpl}}}

    Let us assume as in~\Cref{def:composition} that $\tuple v^1=\tuple u^2$ and that $\phi_1,\phi_2$ do not share any further variables. Let $V_1$ be the set of variables of $\phi_1$, let $V_2$ be the set of variables of $\phi_2$, and let $V$ be the set of variables of $\phi$.
    We will first prove the last sentence of the first part of~\Cref{lemma:kimpl} about $O_1 O_3$-mappings.
    To this end, let $O_1\subseteq \proj_{\tuple u^1}(\phi_1^{\sA}),O_2\subseteq \proj_{\tuple v^1}(\phi_1^{\sA}),O_3\subseteq \proj_{\tuple v^2}(\phi_2^{\sA})$ be orbits under $\gG$, and suppose that $\phi_1^{\sA}$ contains an $O_1 O_2$-mapping $f$ and $\phi_2^{\sA}$ contains an $O_2 O_3$-mapping $g$. Using that $f|_{V_1\cap V_2}$ is contained in the same orbit under $\gG$ as $g|_{V_1\cap V_2}$, find a mapping $h\colon V_1\cup V_2\rightarrow A$ such that $h|_{V_1}$ is contained in the same orbit under $\gG$ as $f$ and $h|_{V_2}$ is contained in the same orbit as $g$. It follows that $h\in(\phi_1\wedge \phi_2)^{\sA}$, and hence $h|_{V}\in \phi^{\sA}$ is an $O_1 O_3$ mapping.
    On the other hand, if $\phi^{\sA}$ contains an $O_1 O_3$-mapping $h'$, we can extend it to a mapping $h\in (\phi_1\wedge \phi_2)^{\sA}$ such that $f:=h|_{V_1}\in \phi_1^{\sA}$ and $g:=h|_{V_2}\in \phi_2^{\sA}$ by the definition of $\phi$. Setting $O_2$ to be the orbit of $h(\tuple v^1)$, we get that $f\in \phi_1^{\sA}$ is an $O_1 O_2$-mapping and $g\in\phi_2^{\sA}$ is an $O_2 O_3$-mapping as desired.
    
    Observe now that the fact that $\phi$ is a $(C,\tuple u^1,E,\tuple v^2)$-pre-implication in $\sA$ follows from the previous paragraph. Indeed, it follows immediately that $\phi$ satisfies items (4) and (5) of~\Cref{def:implication}. To see that items (2) and (3) are satisfied as well
    , take any $g\in\phi_2^{\sA}$ with $g(\tuple v^2)\notin E$, let $O_2$ be the orbit of $g(\tuple u^2)$ under $\gG$, and let $O_3$ be the orbit of $g(\tuple v^2)$. Since $\proj_{\tuple v^1}(\phi_1)= \proj_{\tuple u^2}(\phi_2)$, we can find an $O_1 O_2$-mapping in $\phi_1^{\sA}$, and $\phi^{\sA}$ contains an $O_1 O_3$-mapping witnessing that $C\subsetneq \proj_{\tuple u^1}(\phi^{\sA}), E \subsetneq \proj_{\tuple v^2}(\phi^{\sA})$.

    To prove the second part of the lemma, we will prove that for all orbits $O_1\subseteq C, O_2\subseteq D, O_3\subseteq E$ under $\fG$ such that $\phi_1^{\sA}$ contains an injective $O_1 O_2$-mapping $f$ and $\phi_2^{\sA}$ contains an injective $O_2 O_3$-mapping $g$, $\phi^{\sA}$ contains an injective $O_1 O_3$-mapping $h$.
    Note that as in the previous paragraph, it is enough to find an injective mapping $h\in(\phi_1\wedge \phi_2)^{\sA}$ such that $h|_{V_1}$ is contained in the same orbit under $\gG$ as $f$ and $h|_{V_2}$ is contained in the same orbit as $g$. 
    Let $U_1$ be the set of all injective tuples of variables from $V_1$ of length at most $k$, and for every $\tuple v\in U_1$, let us denote the orbit of $f(\tuple v)$ by $O^f_{\tuple v}$. Similarly, let $U_2$ be the set of all injective tuples of variables from $V_2$ of length at most $k$, and for every tuple $\tuple v\in U_2$, let $O^g_{\tuple v}$ be the orbit of $g(\tuple v)$. 
    Let us define a formula $\psi$ with variables from $V_1\cup V_2$ by $$\psi\equiv\bigwedge\limits_{\tuple v\in U_1}O^{f}_{\tuple v}(\tuple v)\wedge \bigwedge\limits_{\tuple v\in U_2}O^{g}_{\tuple v}(\tuple v).$$ 
    Note that since $\sA$ is a first-order expansion of the canonical $k$-ary structure of $\gG$, all orbits $O^f_{\tuple v}, O^g_{\tuple v}$ are pp-definable from $\sA$, and hence $\psi$ is equivalent to a pp-formula over $\sA$. Since $\gG$ is $k$-neoliberal, and since $\sA$ is preserved by a chain of quasi J\'{o}nsson operations, we can proceed as in the proof of~\Cref{cor:noequality} and use~\Cref{lemma:l+1=rel,lemma:l+2=rel} to show that $\psi^{\sA}$ contains an injective mapping $h$. By the construction and by the $k$-homogeneity of $\gG$, this mapping satisfies our assumptions.
$\qed$

\subsection{Proof of \texorpdfstring{\Cref{observation:impl_properties}}{Observation~\ref{observation:impl_properties}}}

    Let us rename the variables of $\phi_2$ as in~\Cref{def:composition} so that $\tuple v^1=\tuple u^2$, and $\phi_1$ and $\phi_2$ do not share any further variables.
    
    For (1), observe that the number of variables of a $(C,\tuple u,C,\tuple v)$-pre-implication is equal to $2k-|\scope(\tuple u)\cap \scope(\tuple v)|$. By~\Cref{lemma:kimpl}, $\phi$ is a $(C,\tuple u_1,C,\tuple v_2)$-pre-implication, whence    $p=2k-|\scope(\tuple u^1)\cap \scope(\tuple v^2)|$, and since $\scope(\tuple u^1)\cap \scope(\tuple v^2)$ is contained both in $\scope(\tuple u^1)\cap \scope(\tuple v^1)$ and in $\scope(\tuple u^2)\cap \scope(\tuple v^2)$, (1) follows by applying the same reasoning to $p_1$ and $p_2$.

    For (2), observe that by the previous paragraph, $p=p_1$ if, and only if, $\scope(\tuple u^1)\cap\scope(\tuple v^2) = \scope(\tuple u^1)\cap\scope(\tuple v^1)$. Similarly, $p=p_2$ if, and only if, $\scope(\tuple u^1)\cap\scope(\tuple v^2) = \scope(\tuple u^2)\cap\scope(\tuple v^2)$, and (2) follows by the fact that $\tuple v^1=\tuple u^2$.
    
    For (3), suppose that $v^1_i\in \scope(\tuple u^1)$. It follows by (2) that $v^1_i=u^2_i\in\scope(\tuple v^2)$, and since $\phi_2$ was obtained by renaming variables of $\phi_1$, and in particular, $\tuple u^2$ was obtained by renaming $\tuple v^1$, it follows that $u^1_i\in\scope(\tuple v^1)$.$\qed$

\subsection{Proof of \texorpdfstring{\Cref{observation:sinksource}}{Observation~\ref{observation:sinksource}}}

    The second part of the observation is immediate.
    To prove the first part, observe that since $\phi$ is a $(C,C)$-implication in $\sA$, it follows that $\mathbf{Vert}(C)$ is a sink in $\mathcal{B}_{\phi}$. Using the oligomorphicity of $\gG$ and the definition of a $(C,C)$-implication, we get that the induced subgraph of $\mathcal B_\phi$ on $\mathbf{Vert}(C)$ is finite and smooth. Hence, there exists a strongly connected component $S_1\subseteq \mathbf{Vert}(C)$ in $\mathcal{B}_\phi$ which is a sink in the induced subgraph, and hence also in $\mathcal B_\phi$.
    Similarly, one sees that $\mathbf{Vert}(E\backslash C)$ is a source in $\mathcal{B}_{\phi}$, and it contains a strongly connected component $S_2$ which is itself a source in $\mathcal B_{\phi}$.
$\qed$

\subsection{Proof of~\Cref{lemma:kcomplete}}

We will construct the desired complete implication as a conjunction of a power of $\phi$ and $I_\ell^{B}$, where 
$\ell$ is the number of variables of the power of $\phi$. 
Note that for every $n\geq 1$, the number of variables of $\phi^{\circ n}$ is at most $2k$ by~\Cref{def:composition} and this number never decreases with increasing $n$ by item (1) in~\Cref{observation:impl_properties}. 
Hence, there is $n\geq 1$ such that the number of variables of $\phi^{\circ n}$ is the biggest among all choices of $n$. Let us denote the number of variables of $\phi^{\circ n}$ by $\ell$; it follows that the number of variables of $(\phi^{\circ n})^{\circ m}$ is equal to $\ell$ for every $m\geq 1$. Let us replace $\phi$ by $\phi^{\circ n}$. It follows from~\Cref{cor:kimpl} that for every $m\geq 1$, $\phi^{\circ m}\wedge I_\ell^{B}$ is an injective $(C,C)$-implication. Now, it follows by item (3) in \Cref{observation:impl_properties} that there exists a unique bijection $\sigma\colon I_\phi\rightarrow I_\phi$ such that $u_i=v_{\sigma(i)}$ for every $i\in I_\phi$.
Replacing $\phi$ with a power of $\phi$ again, we can assume that $\sigma$ is the identity.

Now, we can find $m\geq 1$ such that the number of strongly connected components  of $\phi^{\circ m}$ is maximal among all possible choices of $m$. It follows that composing $\phi^{\circ m}$ with itself arbitrarily many times does not disconnect any vertices from $\mathbf{Vert}(C)$ which are contained in the same strongly connected component of $\mathcal{B}_{\phi^{\circ m}}$; we replace $\phi$ by $\phi^{\circ m}$. Taking another power of $\phi$ and replacing $\phi$ again, we can assume that every strongly connected component of $\mathcal B_\phi$ contains all loops. Now, setting $p$ to be the number of vertices of $\mathcal B_\phi$, we have that, replacing $\phi$ with $\phi^{\circ p}$, every strongly connected component of $\mathcal B_\phi$ contains all arcs, whence $\phi\wedge I_\ell^{B}$ is a complete injective $(C,C)$-implication.$\qed$

\subsection{Proof of \texorpdfstring{\Cref{lemma:diseqJonsson}}{Lemma~\ref{lemma:diseqJonsson}}}

Before we turn to the proof of~\Cref{lemma:diseqJonsson} we need some auxiliary results.

\begin{lemma}\label{lemma:full_complete-impl}
    Let $\sB$ be a $2$-neoliberal structure with  finite duality and  $\relstr A$ a first-order expansion of $\sB$ which pp-defines a quaternary $(C',\tuple u,C',\tuple v)$-implication $\rho$ where the binary  relation $C' \neq (E \cap I^B_2)$ is injective, $\proj_{\tuple u} \rho = \proj_{\tuple v} \rho = E$,   
    and:
    \begin{enumerate}
    \item every $OP$-mapping with injective $O \subseteq C'$ and injective $P \subseteq C'$ 
     may be chosen to be injective, and
     \item for every injective $P \nsubseteq C'$ there exist an injective $O \nsubseteq C'$ 
     such that $\rho$ contains an injective $OP$-mapping.  
    \end{enumerate}

    Then $\rho$ pp-defines 
    a complete quaternary $(C,\tuple u,C,\tuple v)$-implication $\phi$ with binary and injective $C$  and such that it contains every injective $f \in \rho^{\sA} $ with $f(\tuple u)  \in C$ and $f(\tuple v)  \in C$ as well as  every injective $f$ with $f(\tuple u) \in O_D$ and $f(\tuple v) \in O_D$ for some  injective orbital $O_D \nsubseteq C$.
\end{lemma}

\begin{proof}
Consider $\rho^{\circ m}$ where $m$ is the number of different orbitals in $\sB$. It is easily seen that $\psi := \rho^{\circ m}$ is a complete $(C',C')$-implication. By~\Cref{lemma:kimpl}, we have that $\psi$ satisfies Item~1 and Item~2 with $\rho$ replaced by $\psi$. Since $C' \neq (E \cap I^B_2)$, the latter implies that there is an injective orbital $O_D \nsubseteq C'$ such that $\psi$ contains an injective $O_D O_D$-mapping.

\Cref{observation:sinksource} yields that there exist $C \subseteq C'$
such that $\mathbf{Vert}(C)$ is a strongly connected component which is a sink in $\mathcal{B}_{\psi}$. 
Observe that since $\sA$ is a first-order expansion of $\sB$ and since $\psi$ is complete, $C$ is pp-definable from $\sA$. Indeed, for any fixed orbit $O \subseteq C$ under $\Aut(\sB)$, $C$ is equal to the set of all orbits $P\subseteq C'$ such that $\psi^{\sA}$ contains an $O P$-mapping. 
Moreover, $\psi$ is easily seen to be a complete $(C,C)$-implication in $\sA$.
Furthermore, it satisifes Item~1 from the formulation of the lemma with $C'$ replaced with $C$ and contains an injective $O_D O_D$-mapping.

Since $\sB$ has finite duality, there exists a number $d \geq 3$ such that for every finite structure $\sX$ in the signature of $\sB$, it holds that if every substructure of $\sX$ of size at most $d-2$ maps homomorphically to $\sB$, then so does $\sX$. 
Set $\phi:=\psi^{\circ d}$. Let $V$ be the set of variables of $\phi$. It follows from~\Cref{lemma:kimpl} and the reasoning above that $\phi$ is
a complete $(C,\tuple u,C,\tuple v)$-implication in $\sA$ for some $\tuple u,\tuple v$,
that it again satisfies Item~1 from the formulation of the lemma with $C'$ replaced with $C$ and contains an injective $O_D O_D$-mapping.

We now proceed as in the proof of~\Cref{cor:critical}. 
Let $f_0\in A^V$  
be an assignment such that either $f_0(\tuple u), f_0(\tuple v) \in C$ or $f_0(\tuple u), f_0(\tuple v) \in O_D$.
Up to renaming variables, we can assume that $\psi$ is a $(C,\tuple u,C,\tuple v)$-implication. 
Let $\tuple w^0,\ldots,\tuple w^d$ be disjoint pairs of variables. 
Assume that $\tuple w^0=\tuple u$, $\tuple w^d=\tuple v$. 
For every $j\in\{0,\ldots,d-1\}$, let $\psi_{j+1}$ be the $(C,\tuple w^j,C,\tuple w^{j+1})$-implication obtained from $\psi$ by renaming $\tuple u$ to $\tuple w^j$ and $\tuple v$ to $\tuple w^{j+1}$. 
It follows that $\phi$ is equivalent to the formula obtained from $\psi_1\wedge\dots\wedge \psi_d$ by existentially quantifying all variables that are not contained in $\scope(\tuple u)\cup\scope(\tuple v)$. 
In order to prove our lemma, it is therefore enough to show that $f_0$ can be extended to a mapping $h\in (\psi_1\wedge\dots\wedge\psi_d)^{\sA}$.


Let $O$ be the orbit of $f_0(\tuple u)$ under $\Aut(\sB)$, and let $P$ be the orbit of $f_0(\tuple v)$. Let $f_1\in \psi_1^{\sA}$ be an injective $O P$-mapping, and for every $j\in\{2,\ldots,d\}$, let $f_j\in \psi_j^{\sA}$ be an injective $P P$-mapping. 
Note that such $f_j$ exists for every $j\in[d]$ since $\psi_j$ is complete and since $\mathbf{Vert}(C)$
is a strongly connected component in $\mathcal{B}_{\psi_j}$. 

Let $\tau$ be the signature of $\sB$. Let $X:=\scope(\tuple w^0)\cup\cdots\cup \scope(\tuple w^d)$, and let us define a $\tau$-structure $\sX$ on $X$ as follows. Recall that the relations from $\tau$ correspond to the orbits of injective pairs under $\Aut(\sB)$. For every relational symbol $R\in\tau$, we define $R^{\sX}$ to be the set of all tuples $\tuple w \in I^X_2$ such that there exists $j\in \{0,\ldots,d\}$ such that $\scope(\tuple w)\subseteq \scope(\tuple w^j)\cup\scope(\tuple w^{j+1})$ and $f_j(\tuple w)\in R$; here and in the following, the addition $+$ on indices is understood modulo $d+1$. 
We will show that $\sX$ has a homomorphism $h$ to $\sB$. If this is the case, it follows by the construction and by the $k$-homogeneity of $\Aut(\sB)$ that $h\in (\psi_1\wedge\dots\wedge\psi_d)^{\sA}$. Moreover, we can assume that $h|_{\scope(\tuple w^0)\cup\scope(\tuple w^d)}=f_0$ as desired and as in the proof of~\Cref{lemma:kimpl} we may argue that there exists $h$ which sends elements in $X$  to pairwise different elements in $A$.

Let now $Y\subseteq X$ be of size at most $d-2$. 
Then for cardinality reasons there exists $j\in[d-1]$ such that $Y\subseteq \scope(\tuple w^0)\cup\cdots\cup \scope(\tuple w^{j-1})\cup \scope(\tuple w^{j+1})\cup\cdots\cup \scope(\tuple w^{d})$. 
Observe that $f_{j+1}$ is a homomorphism from the induced substructure of $\sX$ on $\scope(\tuple w^{j+1})\cup \scope(\tuple w^{j+2})$ to $\sA$. 
Since the orbits of $f_{j+1}(\tuple w^{j+2})$ and $f_{j+2}(\tuple w^{j+2})$ agree by definition, by composing $f_{j+2}$ with an element of $\Aut(\sB)$ we can assume that $f_{j+1}(\tuple w^{j+2})=f_{j+2}(\tuple w^{j+2})$. 
We can proceed inductively and extend $f_{j+1}$ 
to $\scope(\tuple w^0)\cup\cdots\cup \scope(\tuple w^{j-1})\cup \scope(\tuple w^{j+1})\cup\cdots\cup \scope(\tuple w^{d})$ such that it is a homomorphism from the induced substructure of $\sX$ on this set to $\sA$. It follows that the substructure of $\sX$ induced on $Y$ maps homomorphically to $\sA$. Finite duality of $\sB$ yields that $\sX$ has a homomorphism to $\sA$ as desired. 
\end{proof}

\begin{lemma}
\label{lemma:4noJonsson}
Let $\sB$ be a $2$-neoliberal structure with finite duality and $\sA$ a first-order expansion of $\sB$ which pp-defines a complete quaternary
 $(C,\tuple u,C,\tuple v)$-implication $\phi$ with binary injective $C$, 
$\proj_{\tuple u} \rho = \proj_{\tuple v} \rho = E$ and such that $\phi^{\sA}$ contains 
    every injective $f \in A^{\scope(\tuple u)\cup \scope(\tuple v)}$ with $f(\tuple u) \in C$ and $f(\tuple v) \in C$ as well as every 
    injective $f \in A^{\scope(\tuple u)\cup \scope(\tuple v)}$ with $f(\tuple u) \in O_D$ and $f(\tuple v) \in O_D$ for some injective $O_D \nsubseteq C$. 
Then  $\sA$ is not preserved by any chain of quasi J\'{o}nsson operations. 
\end{lemma}

\begin{proof}
Suppose on the contrary that $\rho$ is preserved by such a chain $J_1, \ldots, J_{2n+1}$.
   Our goal is to show that there exist two constant vectors $\tuple v^3, \tuple v^4$ with $(v_i^3,v_i^4) \in O_D$ for all $i\in [3]$ and such that  $(J_{2n+1}(\tuple v_3), J_{2n+1}(\tuple v_4)) \in C$. 
   We will say that vectors $\tuple v^3, \tuple v^4$ of which $\tuple v^3$ is constant are of shape $(C, C, O_D)$ if $(v^3_1, v^4_1) = (v^3_2, v^4_2) \in C$ and $(v^3_3, v^4_3) \in O_D$. Similarly, the vectors are of shape $(C, O_D, O_D)$ if $(v^3_2, v^4_2) = (v^3_3, v^4_3) \in O_D$ and $(v^3_1, v^4_1) \in C$. 
   
   We will now prove the following claim.

   \begin{claim}
   For all $i \in [2n+1]$ and all vectors $\tuple v^3,  \tuple v^4$ of shape $(C, C, O_D)$ or $(C, O_D, O_D)$ we have $(J_{i}(\tuple v^3), J_{i}(\tuple v^4)) \in C$.
   \end{claim}

    \begin{proof}
    The proof goes by the induction on $i \in [2n+1]$. For the base case, let $i=1$ and $\tuple v^3, \tuple v^4$  be vectors of shape $(C, C, O_D)$. By $1$-transitivity of $\sB$ there exists $a_3$ such that 
    $\tuple v^3, \tuple q^4$ with $\tuple q^4 = (v^4_1, v^4_2, a_3)$ satisfies $(v^3_i, q^4_i) \in C$ for $i \in [3]$. Since $J_1$ preserves $C$ and by~(\ref{eq:D1}) in~\Cref{def:QuasiJonsson}, for that shape the base case of the induction follows. For $\tuple v^3, \tuple v^4$ of shape $(C, O_D, O_D)$, we intend to find $\tuple v^1, \tuple v^2$ of shape $(C, C, O_D)$
    such that $(v^1_i, v^2_i, v^3_i, v^4_i) \in R_{\phi}$ for $i \in \{1,3\}$ where $R_{\phi}$ contains the tuples from  the image of $\phi^{\sA}$.
    To this end set the constant vector $\tuple v^1$ to some value which is neither in $\tuple v^3$ nor in $\tuple v^4$. 
    Since $\sB$ is $1$-transitive there exists 
    $v^2_2$ such that $(v^1_2, v^2_2) \in C$. Because it has no $2$-algebraicity we may assume that $v^2_2$ is a fresh value, not chosen before. Set $v^2_1 := v^2_2$. Observe that $(v^1_1, v^2_1, v^3_1, v^4_1)$
    is an injective $CC$-mapping and hence it is in $R_{\phi}$.
    Finally, by $1$-transitivity and since $\sB$ has no $2$-algebraicity there exists $v^2_3$ such that $(v^1_3, v^2_3, v^3_3, v^4_3)$ is an injective $O_D O_D$-mapping and hence it is in $R_{\phi}$. Observe that the constructed $\tuple v_1, \tuple v_2$ are of shape $(C,C,O_D)$, and hence $(J_1( \tuple v^1), J_1( \tuple v^2)) \in C$. Now, since $\sB$ is $1$-transitive and has no $2$-algebraicity, there exists $q_1$ which is a fresh value and such that 
    $(q_1, v^2_2, v^3_2, v^4_2) \in R_{\phi}$. Set $\tuple q^1 = ( v^1_1, q_1, v^1_3)$. By~(\ref{eq:Ji}) we have that $(J_1(\tuple q^1), J_1(\tuple v^2)) = (J_1(\tuple v^1), J_1(\tuple v^2)) \in C$.
    Since $\phi$ is a $(C,C)$-implication, it follows that $(J_1(\tuple v^3), J_1(\tuple v^4)) \in C$, which was to be proved. It completes the basic case of the induction.

    For the induction step we assume that the claim holds for $i \in [2n]$ and need to show that it does for $i+1$. The proof depends on whether $i+1$ is odd or even. Assume the former. Then the claim for vectors of shape $(C,C, O_D)$ follows by~(\ref{eq:J2i2i+1}). For the other shape of vectors we just use the same reasoning as in the previous paragraph with a difference that we replace $J_1$ with $J_{i+1}$. Turn now to the case where $i+1$ is even. For the vectors $\tuple v^3 \tuple v^4$ of shape $(C, O_D, O_D)$ we are done by~(\ref{eq:J2i-12i}).
    Assume now that $\tuple v^3, \tuple v^4$ are of shape $(C, C, O_D)$. We now proceed as in the base case but  in the reversed  direction, i.e., we first find $\tuple v^1 \tuple v^2$ of shape $(C, O_D, O_D)$ such that  $(v^1_i, v^2_i, v^3_i, v^4_i) \in R_{\phi}$ for $i \in \{1,3\}$ and then find $\tuple q^1 = (v^1_1, q_1, v^1_3)$ such that $(q_1, v^2_2,  v^3_2, v^4_2) \in R_{\phi}$. By the assumption and~(\ref{eq:Ji}), we have  $(J_{i+1}(\tuple v^1), J_{i+1}(\tuple v^2)) =(J_{i+1}(\tuple q^1), J_{i+1}(\tuple v^2)) \in C$. Then by the fact that $\phi$ is a $(C,C)$-implication, we obtain that $(J_{i+1}(\tuple v^3), J_{i+1}(\tuple v^4)) \in C$. It completes the proof of the claim. 
    \end{proof}

    In order to complete the proof of the lemma we consider  $\tuple v^3, \tuple v^4$ such that $(v^3_i, v^4_i) \in O_D$ for $i \in [3]$. Set $\tuple v^1$ to be any constant vector with fresh values, i.e., with values that do not occur in  
     $\tuple v^3, \tuple v^4$. Since $\sB$ is $1$-transitive and has no $2$-algebraicity, we can find a fresh value $v^2_1$ such that $(v^1_1, v^2_1) \in C$ and a fresh value $v^2_2 = v^3_2$ such that $(v^1_i, v^2_i, v^3_i, v^4_i)$ 
     for $i \in \{ 2,3 \}$ is an injective $O_D O_D$-mapping, and hence in
     $R_{\phi}$. Observe that $\tuple v^1, \tuple v^2$ are of shape 
     $(C, O_D, O_D)$ and hence by the claim above we have 
     $(J_{2n+1}(\tuple v^1), J_{2n+1}(\tuple v^2)) \in C$.
     Since $\sB$ is $1$-transitive and has no $2$-algebraicity we now find 
     $\tuple q^1 = (q^1, w^1_2, w^1_3)$ with a fresh value $q_1$ such that $( q^1_1,  v^2_1, v^3_1, v^4_1) $ is an injective $O_D O_D$-mapping and hence in $R_{\phi}$. By~(\ref{eq:J2n+1}) we have that $(J_{2n+1}(\tuple q^1), J_{2n+1}(\tuple v^2)) = (J_{2n+1}(\tuple v^1), J_{2n+1}(\tuple v^2)) \in C$.
     Since $\phi$ is a $(C,C)$-implication we have that $(J_{2n+1}(\tuple v^3), J_{2n+1}(\tuple v^4)) \in C$ which yields the contradiction since $O_D$ is pp-definable in $\sA$ and  $O_D \nsubseteq C$.
 \end{proof}  

\begin{lemma}
\label{lemma:2ternaryJonsson}
    Let $\sB$ be a $2$-neoliberal structure with finite duality and $\sA$ a first-order expansion of $\sB$ which pp-defines  
    a complete $(C, (x_1,x_3), C, (x_2, x_3))$-implication $\phi$ with $C$ injective, 
    $\proj_{(x_1, x_3)} \phi = \proj_{(x_2, x_3)} \phi = E$ 
    and $O_D$ an orbital  in $E \setminus C$. Suppose that $\phi^{\sA}$
    contains all $f$ satisfying both $f(x_1, x_3) \in C$ and $f(x_2, x_3) \in C$ as well as both $f(x_1, x_3) \in O_D$ and $f(x_2, x_3) \in O_D$.
    
    If there exist orbitals $O_1, O_2 \subseteq C$ such that $$\rho(x_1, x_2, x_3) \equiv (\exists z~\phi(x_1, x_2, z) \wedge E(x_1, x_3) \wedge E(x_2, x_3))$$ is satisfied by $f$ with $f(x_1, x_3) \in O_1$ and
    $f(x_2, x_3) \in O_D$ as well as by $f$ with $f(x_1, x_3) \in O_D$ and
    $f(x_2, x_3) \in O_2$, then 
    $\sA$ is not preserved by any chain of quasi J\'{o}nsson operations.
\end{lemma} 

    \begin{proof}
    Suppose that there exists a chain of quasi J\'{o}nsson operations $J_1, \ldots, J_{2n+1}$. 
    Let $\tuple w^3$ be a constant vector. 
    We say that  vectors $\tuple w^2, \tuple w^3 \in A^3$ are of shape $(O_l,O_l,O_D)$ with $l \in [2]$ if  $(w^2_i, w^3_i) \in O_l$ for $i \in [2]$ and 
    $(w^2_3, w^3_3) \in O_D$; and that they are of shape 
    $(O_l,O_D,O_D)$ with $l \in [2]$ if $(w^2_i, w^3_i) \in O_D$ for $i \in \{ 2,3 \}$ and 
    $(w^2_1, w^3_1) \in O_l$ with $l \in [2]$. 
    Our goal is to show that the assumption on the existence of $J_1, \ldots, J_{2n+1}$ implies that there exist constant vectors $\tuple w^2, \tuple w^3$ with 
    $( w^2_i, w^3_i) \in O_D$ such that $(J_{2n+1}(\tuple w^2), J_{2n+1}(\tuple w^3)) \in C$. It yields the contradiction since $O_D$ is pp-definable in $\sA$ and $O_D \nsubseteq C$. As usual we start with a claim.

    \begin{claim}
    For all $i\in [2n+1]$ and all vectors $\tuple v^2, \tuple v^3$ 
     of shape $(O_l,O_l,O_D)$ and 
    $(O_l,O_D,O_D)$ with $l \in [2]$ we have $(J_i(\tuple v^2),  J_i(\tuple v^3))$.
    \end{claim} 

    \begin{proof}
    The proof goes by induction on  $i \in [2n+1]$. For the base case ($i = 1$) and vectors of shape $(O_l,O_l,O_D)$ with $l \in [2]$ the proof follows as usual by~(\ref{eq:J1}) in~\Cref{def:QuasiJonsson} and the fact that $C$ is pp-definable in $\sA$. 
    Before we turn to the proof of the base case for the vectors of the other shape, we observe that the set $\phi^{\sA}$ can be viewed as a ternary relation $R(x_1,x_2, x_3)$.
    
    For $\tuple w^2, \tuple w^3$ of  shape $(O_1, O_D, O_D)$, we need $\tuple w^1$ 
    such that $\tuple w^1, \tuple w^3$ is of shape $(O_1,O_1,O_D)$, $(w^1_i, w^2_i, w^3_i) \in R$ for $i \in \{ 1,3 \}$ and $(w^1_2,w^2_2)$
    $ \in \proj_{(x_1, x_2)}(\phi^{\sA})$. Since $(\rho^{\sA})$ contains an $O_1 O_D$-mapping, we may find the appropriate $w^1_2$. Then $(w^1_1, w^2_1, w^3_1) \in R$ with $w^1_1 = w^1_2$ since by the assumption of the lemma every $O_1 O_1$-mapping is in $\phi^{\sA}$.
    Now, if $O_D$ is $=$, then we just set $w^1_3 = w^3_3$. Otherwise we use 
  $1$-transitivity and the assumption on $\phi$   to find $w^1_3$ satisfying $(w^1_3, w^3_3) \in P$ and 
 $(w^1_3, w^2_3, w^3_3) \in R$. 
  The vectors $\tuple w^1, \tuple w^3$ are of shape $(O_1, O_1, O_D)$, and hence $(J_1(\tuple w^1), J_1(\tuple w^3)) \in C$. Since 
$(w^1_2,w^2_2) \in \proj_{(x_1, x_2)}(\phi^{\sA})$, there exists $a_3$ such that
$(w^1_2, w^2_2,a_3) \in \phi^{\sA}$. Let $\tuple q^3 = (w^3_1, a_3, w^2_3)$. By~(\ref{eq:Ji}), we have that $(J_1(\tuple w^1), J_1(\tuple w^3)) = 
    (J_1(\tuple w^1), J_1(\tuple q^3)) \in C$. Since $\phi$ is a 
    $(C, C)$-implication, it follows that
    $(J_1(\tuple w^2), J_1(\tuple q^3)) = 
    (J_1(\tuple w^2), J_1(\tuple w^3)) \in C$. 
    In order to complete the base case we need to consider vectors $\tuple w_2, \tuple w_3$ of shape $(O_2, O_D, O_D)$. Since $\phi^{\sA}$ contains a $O_1 O_2$-mapping and a $O_D O_D$-mapping, it follows that  there exists $\tuple w^1$ such that $\tuple w^1, \tuple w^3$ are of shape $(O_1, O_D, O_D)$ and such that 
    $(w^1_i, w^2_i, w^3_i)$ is in $R$ for $i \in [3]$. Since  
    $(J_1(\tuple w^1), J_1(\tuple w^3)) \in C$ and 
    $\phi$ is a $(C,  C)$-implication, it follows that
    $(J_1(\tuple w^2), J_1(\tuple w^3)) \in C$, which was to be proved. It completes the base case. 

    For the induction step we consider the case where $i$ is odd and where it is even. In the former case, we are done for vectors of shape $(O_l, O_l, O_D)$ with $l \in [2]$ by~(\ref{eq:J2i2i+1}) and 
    the step from these vectors to the vectors of shape 
    $(O_l, O_D, O_D)$ is carried out as in the base case with a difference that we replace $J_1$ with $J_i$.  In the latter case we are done for the shape $(O_l, O_D, O_D)$ by~(\ref{eq:J2i-12i}). For the other shapes  consider $\tuple w^2, \tuple w^3$ of shape $(O_2,O_2, O_D)$. We need to show that there exists $\tuple w^1$ such that $\tuple w^1, \tuple w^3$ is of shape
    $(O_2,O_D, O_D)$,  $(w^1_i, w^2_i, w^3_i) \in R$ for $i \in \{ 1,3 \}$ and $(w^1_2,w^2_2) \in \proj_{(x_1, x_2)}(\phi^{\sA})$. Since $\rho^{\sB}$ contains a $O_D O_2$-mapping, we may find the appropriate $w^1_2$. Then $(w^1_3, w^2_3, w^3_3) \in R$ for $w^1_3 := w^1_2$ by the assumption that every $O_D O_D$-mapping is in $\rho^{\sA}$.  Finally, by $1$-transitivity of $\sB$ and the assumption that every $O_2 O_2$-mapping is in $\phi^{\sB}$ we have $w^1_1$ satisfying $(w^1_1, w^2_1) \in O_2$ and 
    $(w^1_1, w^2_1, w^3_1) \in R$. 
    Notice that $(J_i(\tuple w^1), J_i(\tuple w^2)) \in C$. Since 
    $(w^1_2,w^2_2) \in \proj_{(x_1, x_2)}(\phi^{\sA})$, there exists $a_3$ such that
    $(w^1_2, w^2_2,a_3) \in \phi^{\sA}$. Let $\tuple q^3 = (w^3_1, a_3, w^2_3)$. By~(\ref{eq:Ji}), we have that $(J_i(\tuple w^1), J_i(\tuple w^3)) = 
    (J_i(\tuple w^1), J_i(\tuple q^3)) \in C$. Since $\phi$ is a $(C, C)$-implication, it follows that
    $(J_i(\tuple w^2), J_i(\tuple q^3)) = 
    (J_i(\tuple w^2), J_i(\tuple w^3)) \in C$. 
    In order to complete the induction step we need to consider vectors $\tuple w^2, \tuple w^3$ of shape $(O_1, O_D, O_D)$. By the assumption, $\phi$ contains all $O_2 O_1$-mappings and all  $O_D O_D$-mappings, and hence there exists $\tuple w^1$ such that $\tuple w^1, \tuple w^3$ are of shape $(O_2, O_D, O_D)$ and such that 
    $(w^1_i, w^2_i, w^3_i)$ is in $R$ for $i \in [3]$. Since  
    $(J_i(\tuple w^1), J_i(\tuple w^3)) \in C$ and 
    $\phi$ is a $(C,  C)$-implication, it follows that
    $(J_i(\tuple w^2), J_i(\tuple w^3)) \in C$, which was to be proved. It completes the proof of the induction step, and of the claim. 
\end{proof}

    As usual, we now prove that $(J_{2n+1}(\tuple w^2), J_{2n+1}(\tuple w^3)) \in C$  for constant $\tuple w^2, \tuple w^3$ such that $(w^2_i, w^3_i) \in O_D$ for $i \in [3]$. To this end we need to find 
    $\tuple w^1$ of shape $(O_1, O_D, O_D)$  such that $(w^1_i, w^2_i, w^3_i) \in R$ for $i \in \{2, 3\}$.
    and $(w^1_1, w^1_2) \in \proj_{(x_1, x_2)} (\phi^{\sA})$. The required $w^1_2 = w^1_3$ exists by the fact that $\rho^{\sA}$ contains all 
    $O_D O_D$-mappings.  When it comes to $w^1_1$, we may find it by $1$-transitivity and the fact that  $\rho^{\sA}$ contains an $O_1 O_D$-mapping.  
    Since $(w^1_1, w^2_1) \in \proj_{(x_1, x_2)} (\phi^{\sA})$, there exists 
    $a_3$
    such that $(w^1_1, w^2_1, a_3) \in R$. By~(\ref{eq:J2n+1}) we have that $(J_{2n+1}(\tuple w^1), J_{2n+1}(\tuple q^3)) \in C$ where $\tuple q^3 = (a_3, w^3_2, w^3_3)$. Since the formula $\phi$ is a $(C, (x_1, x_2),  C, (x_1, x_3))$-implication, it follows that $(J_{2n+1}(\tuple w^2), J_{2n+1}(\tuple q^3)) = (J_{2n+1}(\tuple w^2), J_{2n+1}(\tuple w^3)) \in C$. It completes the proof of the lemma.
    \end{proof}

\noindent
We are now ready to prove~\Cref{lemma:diseqJonsson}.

\diseqJonsson*
\begin{proof}
For any orbital $O$ under $\Aut(\sB)$ consider a relation $\phi_O$  over three variables $x_1, x_2, x_3$ 
defined simply by $O(x_1, x_2)$. 
Define now a directed graph $G$ over pp-defnable  subsets of $I^B_2$ such that there is an arrow from $S_1$ to $S_2$ if for some $O$ the formula $\phi_O$ is a $(S_1, (x_1, x_3), S_2,(x_2, x_3))$-implication. Observe that every non-trivial subset $S$ of $I^B_2$ has an outgoing arc in $G$. Indeed, $\phi_O$ for $O \nsubseteq S_1$ is a $(S_1, (x_1, x_3), S_2, (x_2, x_3))$-implication for some $S_2 \subseteq I^B_2$. Indeed, the inclusion holds since $O \nsubseteq S_1$, and hence $S_2$ cannot contain $=$. Clearly $G$ contains all orbitals and hence it is non-empty. Now, either $G$ contains $I^B_2$ and we are done or it contains a cycle from $C'$ back to $C'$ for some $C' \subsetneq I^B_2$. Since, by $1$-transitivity of $\sA$, the projection of every $\phi_O$ to both $(x_1, x_3)$ and $(x_2, x_3)$ is $B^2$, we may, as in the proof of~\Cref{cor:critical}, compose the implications along the cycle. This yields a $(C', (x_1, x_3), C', (x_2, x_3))$-implication $\rho'$ whose projection on both $(x_1, x_3)$ and $(x_2, x_3)$ is $B^2$.

Since $S_1, S_2 \subsetneq I^B_2$ and the automorphism group of $\sB$ has no $2$-algebraicity, we have that every $(S_1, (x_1, x_3), S_2,(x_2, x_3))$-implication $\phi_O$ has the two following properties.
\begin{enumerate}
\item Every $P_1P_2$-mapping in $\phi_O$ with injective $P_1, P_2$  may be chosen to be injective.
\item There exists an injective $P_2 \nsubseteq S_2$ and for every such $P_2$ there exists   
an injective orbital $P_1 \nsubseteq S_1$ such that
$\phi_O$ contains an injective $P_1 P_2$-mapping.
\end{enumerate}

By~\Cref{lemma:kimpl}, we have that the above two properties hold also for 
$\rho'$ where $S_1, S_2$ are replaced by $C'$. Furthermore, by composing $\rho'$ with itself as many times as there are orbitals in $\sA$, as in the proof of \Cref{lemma:kcomplete}, we obtain a complete $(C', (x_1, x_3), C', (x_2, x_3))$-implication which additionally contains an injective $O_D O_D$-mapping with $O_D \notin C$. Finally, mimicking the proof of~\Cref{cor:critical} we obtain a $(C, (x_1, x_3), C, (x_2, x_3))$-implication $\rho''$ with $C \subseteq C'$ pp-definable in $\sA$ which contains an injective $O_D O_D$-mapping with $O_D \notin C$ and which in addition to (1) and (2) where $S_1, S_2$ are replaced with $C$  satisfies also the following.
\begin{enumerate}
\setcounter{enumi}{2}
    \item 
    For the union $U$ of orbitals in every strongly connected component in $\mathcal{B}_{\rho''}$, $\rho''^{\sB}$ contains every $f \in B^{\{x_1, x_2, x_3 \} }$ with $f(x_1, x_3) \in U$ and $f(x_2, x_3) \in U$.
\end{enumerate}
Indeed, if both $f(x_1, x_3) \in P_1$ and $f(x_2, x_3) \in P_2$ with both $P_1$ and $P_2$ injective, then 
we proceed as in the proof of~\Cref{cor:critical} for $P_1, P_2$ both contained in $C$ or both contained in $D$. Otherwise, if $P_2$ is $=$, then 
for every $f(x_1, x_3) \in P_1$, there exists 
only one possible $f(x_2, x_3) \in P_2$, and hence the property follows. The same reasoning holds when $P_1$ is $=$. Hence $\rho''$ has all the three properties: (1)--(3); and also the following one which easily follows from (3).
\begin{enumerate}
\setcounter{enumi}{3}
    \item 
    The relation defined by $\exists x_3~C(x_1,x_3) \wedge C(x_2, x_3)$ is contained in $\proj_{(x_1, x_2)} \rho''^{\sB}$.
\end{enumerate}

We now claim that $G(x,y) \equiv \exists z~\rho''(x,z,y) \wedge \rho''(z,x,y)$ contains exactly these orbitals  that 
   are involved in some strongly connected component of $\mathcal{B}_{\rho''}$. 
Indeed, on one hand assume that a pair $(x,y) \in O$ is such that $O$ is involved in a strongly connected component of $\mathcal{B}_{\rho''}$. 
   By Property~(3) above, we have that there is $z$ such that $(z,y) \in O$ and $\rho''(x,z,y) \wedge \rho''(z,x,y)$ is satisfied. 
   On the other hand, if $O$ is not involved in a strongly connected component of $B_{\rho''}$, then there is no orbital $P$ such that both: there is an arc $(O,P)$ and an arc $(P,O)$ in $\mathcal{B}_{\rho''}$. Thus, there is no $z$ satisfying $\rho''(x,z,y) \wedge \rho''(z,x,y)$. Observe now that 
   $\rho_G(x_1,x_2,x_3) \equiv \rho''(x_1,x_2,x_3) \wedge G(x_1,x_2) \wedge G(x_1, x_3,z)$ 
   is again a $(C, (x_1, x_3), C, (x_2, x_3))$-implication which satisfies all the four properties: (1)--(4) above and all orbitals in $G = \proj_{(x_1, x_3)} (\rho_G^\sA) = \proj_{(x_2, x_3))} (\rho_G^\sA)$ are involved in some strongly connected component of $\mathcal{B}_{\rho_G}$. 
   For the remainder of the proof we assume that $G$ is a minimal wrt. the inclusion pp-definable binary relation which contains $C$ and such that 
   $\rho_G$ defined as above satisfies (1)--(4) and    
   all vertices in $\mathcal{B}_{\rho_G}$   are involved in some strongly connected component of the graph.

Since $C$ is pp-definable in $\sA$, so are the projection $Y_{2,3}$ of 
$(\exists y~\rho_{G}(x_1, x_2, y) \wedge C(x_1, x_3))$ to $(x_2, x_3)$ and the projection $Y_{1,3}$ of 
$(\exists y~\rho_{G}(x_1, x_2, y) \wedge C(x_2, x_3))$
 to $(x_1, x_3)$. By the minimality of $G$ it follows that both $Y_{2,3}$ and $Y_{1,3}$ is either $G$ or $C$. If both are $G$, then there exist orbitals $O_1, O_2 \subseteq C$ and an orbital $O_D \subseteq G \setminus C$ such that 
$$Z(x_1, x_2, x_3) \equiv \exists y~\rho_{G}(x_1, x_2, y) \wedge G(x_1, x_3) \wedge G(x_2, x_3)$$ contains an $O_1 O_D$-mapping and an $O_D O_2$-mapping. 
Since $\rho_G$ satisfies~(3), \Cref{lemma:2ternaryJonsson} asserts that $\sA$ is not preserved by any chain of quasi J\'{o}nsson operations. The lemma follows.
Observe that  in particular $\exists y~\rho_{G}(x_1, x_2, y) \supseteq I^B_2$ implies $Y_{1,3} = Y_{2,3} = G$.

From now on we therefore assume without loss of generality that  $Y_{2,3}$ is $C$ and $\exists y~\rho_{G}(x_1, x_2, y) \nsupseteq I^B_2$. 
By~(4),  we have that $H(x_1, x_2) \equiv \exists x_3~C(x_1, x_3) \wedge C(x_2, x_3)$ is contained in $\exists y~(\rho_G(x_1, x_2, y))$ and in particular $H(x_1, x_2) \nsupseteq I^B_2$. Observe that $\xi_1(x_1, x_2, x_3) \equiv C(x_2, x_3)$
is a $(C, (x_1, x_3), H, (x_1, x_2))$-implication 
and that $\xi_2(x_1, x_2, x_4) \equiv C(x_1, x_4)$
is a $(H, (x_1, x_2), C, (x_2, x_4))$-implication and that the relation $\xi(x_1, x_2, x_3, x_4) \equiv \xi_1 \circ \xi_2$ is a quaternary $(C, (x_1,x_3), C, (x_2, x_4))$-implication. Since the automorphism group of the structure $\sB$ has no $2$-algebraicity,  the projection of $\xi_1$ ($\xi_2$) on $(x_1, x_3)$ and $(x_1, x_2)$  (on $(x_1, x_2)$ and $(x_2, x_4)$) is $B^2$.
In consequence, the projection of $\xi$ on $(x_1, x_2)$ and $(x_2, x_4)$ is $B^2$. 
Since~(2) holds for $\rho_G$ we have that $C \neq G \cap I^B_2$, and in particular that $C \subsetneq I^B_2$. 
We will now argue that
\begin{enumerate} 
\setcounter{enumi}{4}
    \item  every  $OP$-mapping in $\xi$ with $O \subseteq C$ and $P \subseteq C$ may be chosen to be injective and that
    \item for every injective $P \nsubseteq C$ there exists an injective $O \nsubseteq C$ such that $\xi$ contains an injective $OP$-mapping.  
\end{enumerate}
It will offer us an opportunity to use~\Cref{lemma:full_complete-impl,lemma:4noJonsson}, and therefore to complete the proof of the lemma.
Indeed, since the automorphism group of $\sB$ has no $2$-algebraicity, it holds that every $O_1P_1$-mapping in $\xi_1$ (every $O_2 P_2$-mapping in $\xi_2$) with  injective $O_1 \subseteq C$ ($O_2 \subseteq H$) and $P_1 \subseteq H$ ($P_2 \subseteq C$) may be chosen to be injective. Thus, Item~(5) follows by~\Cref{lemma:kimpl}. Similarly, Item~(6) holds for $\xi_1$ and $\xi_2$ by the fact that the automorphism group of $\sB$ has no $2$-algebraicity. Hence~(6) follows by~\Cref{lemma:kimpl}. 
\end{proof}

\subsection{Proof of \texorpdfstring{\Cref{lemma:2critical}}{Lemma~\ref{lemma:2critical}}}

The proof goes along the lines of some parts of the proof of~\Cref{lemma:diseqJonsson}
with a difference that this time, due to the mentioned lemma, we can pp-define $I^B_2$, and hence assume that $\phi$ is injective.  
The formula $\phi$ is critical in $\sA$ over $(C,D, \tuple u, \tuple v)$, the intersection of $\phi$ with $I^B_3$ is injective and again critical in $\sA$ over $(C,D, \tuple u, \tuple v)$. Without loss of generality we assume that $\tuple u = (x_1, x_3)$ and $\tuple v = (x_2, x_3)$.

Now, as in the proof of~\Cref{lemma:diseqJonsson}, we define the binary relation $G(x,y) \equiv \exists z~\phi(x,z,y) \wedge \phi(z,x,y)$
which contains exactly these orbital in $\proj_{\tuple u} \phi^{\sA}$ that 
   are involved in some strongly connected component of $\mathcal{B}_{\phi}$. 
Then we observe that 
   $\phi_G(x_1,x_2,x_3) \equiv \phi(x_1,x_2,x_3) \wedge G(x_1,x_3) \wedge G(x_2,x_3)$ 
   is again critical in $\sA$ over $(C, D, \tuple u, \tuple v)$ and that all orbitals in $G = \proj_{\tuple u} (\phi_G^\sA) = \proj_{\tuple v} (\phi_G^\sA)$ are involved in some strongly connected component of $\mathcal{B}_{\phi_G}$. 
   We also assume for the remainder of the proof that $G$ is minimal wrt. the inclusion pp-definable binary relation which contains $C$ and such that 
   $\phi_G$ is  critical in $\sA$ over $(C, D, \tuple u, \tuple v)$ for some $D$ disjoint from $C$.

By the minimality of $G$, we have that the projection $Y_{2,3}$ of 
$(\exists y~\phi_{G}(x_1, x_2, y) \wedge C(x_1, x_3))$ to $(x_2, x_3)$ and the projection $Y_{1,3}$ of 
$(\exists y~\phi_{G}(x_1, x_2, y) \wedge C(x_2, x_3))$
 to $(x_1, x_3)$ is either $G$ or $C$. If both are $G$, then there exists orbitals $O_1, O_2 \subseteq C$ and an orbital $O_D \subseteq G \setminus C$ such that 
$\rho(x_1, x_2, x_3) \equiv \exists y~\phi_{G}(x_1, x_2, y) \wedge G(x_1, x_3) \wedge G(x_2, x_3)$ contains an $O_1 O_D$-mapping and an $O_D O_2$-mapping. 
Since $\phi_G$ is critical, \Cref{lemma:2ternaryJonsson} asserts that $\sA$ is not preserved by any chain of quasi J\'{o}nsson operations. The lemma follows.
Observe again that  in particular $\exists y~\phi_{G}(x_1, x_2, y) \supseteq I^B_2$ implies $Y_{1,3} = Y_{2,3} = G$.

From now on the proof goes exactly as the one for~\Cref{lemma:diseqJonsson}.  
We assume without loss of generality that  $Y_{2,3}$ is $C$ and that the projection $\exists y~\phi_{G}(x_1, x_2, y) \nsupseteq I^B_2$. 
Since $\phi_G$ is critical we may assume that  $H(x_1, x_2) \equiv \exists x_3~C(x_1, x_3) \wedge C(x_2, x_3)$ is contained in $\proj_{x_1, x_2} \phi_G$.  It follows that $H(x_1, x_2) \nsupseteq I^B_2$. Now we define ternary $\xi_1, \xi_2$ and a quaternary $\xi$ which is a $(C, (x_1, x_3), C(x_2, x_4))$-implication whose projections on $(x_1, x_3)$ and $(x_2, x_4)$ are $B^2$. Hence $C \subsetneq I^B_2$. We also have that (5) and (6) from the proof above hold for $\xi$. Thus, we are done by the applications of~\Cref{lemma:full_complete-impl,lemma:4noJonsson}.

\end{document}